\documentclass[cleveref,autoref,numberwithinsect,thm-restate]{lipics-v2021}

\usepackage{graphicx} 
  \setkeys{Gin}{width=\linewidth,totalheight=\textheight,keepaspectratio}
  \graphicspath{{graphics/}} 
\usepackage{amsmath}
\usepackage{amssymb}  
\usepackage{amsthm}
\usepackage{thmtools}
\usepackage{stmaryrd}
\usepackage[compress]{cite}
\usepackage{complexity}
\usepackage{booktabs} 
\usepackage{units}    
\usepackage{multicol} 
\usepackage{fancyvrb} 
  \fvset{fontsize=\normalsize}
\usepackage{todonotes} 
\usepackage{newfile}
\usepackage{hyperref}

\usetikzlibrary{arrows.meta,decorations.pathmorphing,automata}

\newcommand{\regsep}[2]{#1\mathbin{|}#2}

\DeclareMathOperator{\supp}{supp}

\usepackage{bm}
\newcommand{\Z}{\mathbb{Z}}
\newcommand{\N}{\mathbb{N}}
\newcommand{\bbF}{\mathbb{F}}

\newcommand{\bd}{\bm{d}}
\newcommand{\be}{\bm{e}}
\newcommand{\bmf}{\bm{f}}
\newcommand{\bzero}{\bm{0}}

\newcommand{\bu}{\bm{u}}
\newcommand{\bv}{\bm{v}}

\newcommand{\bp}{\bm{p}}
\newcommand{\bq}{\bm{q}}
\newcommand{\br}{\bm{r}}
\newcommand{\bs}{\bm{s}}
\newcommand{\bx}{\bm{x}}
\newcommand{\by}{\bm{y}}

\newcommand{\bb}{\bm{b}}

\newcommand{\cA}{\mathcal{A}}
\newcommand{\cG}{\mathcal{G}}
\newcommand{\cV}{\mathcal{V}}
\newcommand{\cW}{\mathcal{W}}
\newcommand{\inv}[1]{#1^{-1}}
\newcommand{\restrict}[2]{\ensuremath{\pi_{#2}(#1)}}
\newclass{\NEXPTIME}{NEXPTIME}

\newcommand{\skel}{\operatorname{skel}}
\newcommand{\cycles}{\operatorname{cycles}}

\renewcommand{\Sigma}{\varSigma}
\renewcommand{\Gamma}{\varGamma}

\title{The complexity of separability for semilinear sets and Parikh automata}
\author{Elias Rojas Collins}
{Massachusetts Institute of Technology, Cambridge, MA, USA} 
{erojasc@mit.edu}
{https://orcid.org/0009-0003-3929-1386} 
{} 

\author{Chris K\"{o}cher}
{Max Planck Institute for Software Systems, Kaiserslautern, Germany} 
{ckoecher@mpi-sws.org}
{https://orcid.org/0000-0003-4575-9339} 
{} 

\author{Georg Zetzsche}
{Max Planck Institute for Software Systems, Kaiserslautern, Germany} 
{georg@mpi-sws.org}
{https://orcid.org/0000-0002-6421-4388} 
{} 

\authorrunning{E. Rojas Collins, C. Köcher, and G. Zetzsche}

\Copyright{Elias Rojas Collins, Chris Köcher, and Georg Zetzsche}

\funding{\flag[3cm]{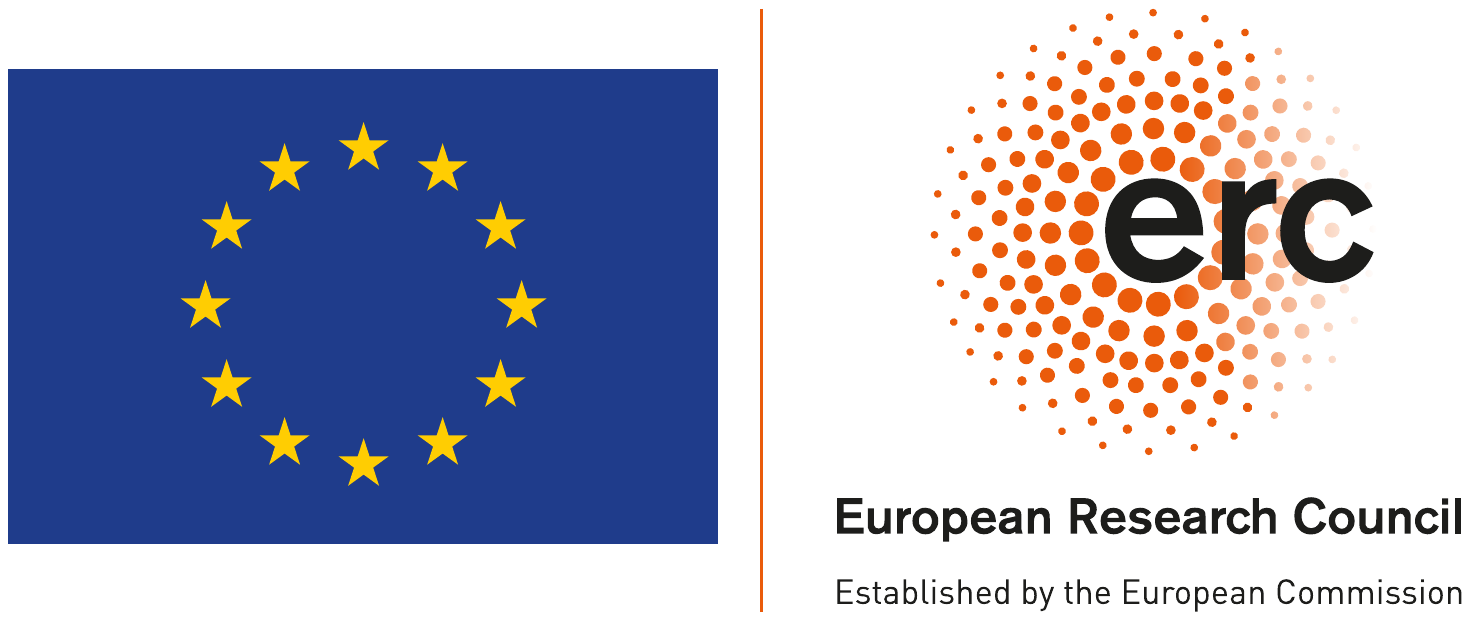}Funded by the European Union (ERC, FINABIS, 101077902). Views and opinions expressed are however those of the author(s) only and do not necessarily reflect those of the European Union or the European Research Council Executive Agency. Neither the European Union nor the granting authority can be held responsible for them.}

\ccsdesc[500]{Theory of computation~Models of computation}
\ccsdesc[300]{Theory of computation~Regular languages}

\keywords{Vector Addition System, Separability, Regular Language}

\relatedversiondetails{Full version}{https://arxiv.org/abs/2505.11199}

\newif\ifFullVersion
\FullVersiontrue

\ifFullVersion\relax\else
  \EventEditors{Pawe\l{} Gawrychowski, Filip Mazowiecki, and Micha\l{} Skrzypczak}
  \EventNoEds{3}
  \EventLongTitle{50th International Symposium on Mathematical Foundations of Computer Science (MFCS 2025)}
  \EventShortTitle{MFCS 2025}
  \EventAcronym{MFCS}
  \EventYear{2025}
  \EventDate{August 25--29, 2025}
  \EventLocation{Warsaw, Poland}
  \EventLogo{}
  \SeriesVolume{345}
  \ArticleNo{29}
\fi

\begin{document}
\nolinenumbers

\maketitle

\begin{abstract}
	In a \emph{separability problem}, we are given two sets $K$ and $L$
	from a class $\mathcal{C}$, and we want to decide whether there exists
	a set $S$ from a class $\mathcal{S}$ such that $K\subseteq S$ and
	$S\cap L=\emptyset$. In this case, we speak of \emph{separability of
	sets in $\mathcal{C}$ by sets in $\mathcal{S}$}.

	We study two types of separability problems. First, we consider
	separability of semilinear sets (i.e.\ subsets of $\N^d$ for some $d$) by sets definable by quantifier-free
  monadic Presburger formulas (or equivalently, the recognizable subsets
  of $\N^d$). Here, a formula is monadic if each atom uses at most one
  variable. Second, we consider separability of languages of Parikh
	automata by regular languages. A Parikh automaton is a machine with
	access to counters that can only be incremented, and have to meet a
	semilinear constraint at the end of the run. Both of these separability
	problems are known to be decidable with elementary complexity.

	Our main results are that both problems are $\coNP$-complete. In the
	case of semilinear sets, $\coNP$-completeness holds regardless of
	whether the input sets are specified by existential Presburger
	formulas, quantifier-free formulas, or semilinear representations.  Our
	results imply that recognizable separability of rational subsets of
	$\Sigma^*\times\N^d$ (shown decidable by Choffrut and Grigorieff) is
	$\coNP$-complete as well.  Another application is that regularity of
	deterministic Parikh automata (where the target set is
	specified using a quantifier-free Presburger formula) is
	$\coNP$-complete as well.  
\end{abstract} 

\section{Introduction}
\subparagraph{Separability}
In a \emph{separability problem}, we are given two sets $K$ and $L$ from a
class $\mathcal{C}$, and we want to decide whether there exists a set $S$ from
a class $\mathcal{S}$ such that $K\subseteq S$ and $S\cap L=\emptyset$. Here,
the sets in $\mathcal{S}$ are the admissible separators, and $S$ is said to
\emph{separate} the sets $K$ and $L$. In the case where $\mathcal{C}$ is a class of non-regular languages and $\mathcal{S}$ is the class of regular
languages, then the problem is called \emph{regular separability (problem) for $\mathcal{C}$}.
While the problem turned out to be undecidable for context-free languages in the 1970s~\cite{SzymanskiW-sicomp76,DBLP:journals/jacm/Hunt82a}, the last decade saw a significant amount of attention on regular separability for subclasses (or variants) of \emph{vector addition systems with states (VASS)}.
Regular separability was studied for coverability languages of VASS (and, more generally, well-structured transition systems)~\cite{DBLP:conf/concur/CzerwinskiLMMKS18,DBLP:conf/concur/Keskin023,DBLP:conf/fsttcs/KocherZ23}, one-counter automata and one-dimensional VASS~\cite{CzerwinskiL17}, Parikh automata~\cite{RegSepParikh}, commutative VASS languages~\cite{SepReach}, concerning its relationship with the intersection problem~\cite{DBLP:conf/fsttcs/ThinniyamZ19}, B\"{u}chi VASS~\cite{DBLP:conf/icalp/0001K0Z24,DBLP:conf/stacs/00010Z23}, and also for settings where one input language is an arbitrary VASS and the other is from some subclass~\cite{LICS20}. 
Recently, this line of work culminated in the breakthrough result that regular separability for general VASS languages is decidable and Ackermann-complete~\cite{DBLP:conf/lics/Keskin024}.
However, for subclasses of VASS languages, the complexity landscape is far from understood.

\subparagraph{Separating Parikh automata} An important example of such a subclass is the
class of languages accepted by \emph{Parikh automata}, which are
non-deterministic automata equipped with counters that can only be incremented.
Here, a run is accepting if the final counter values belong to a particular
semilinear set. Languages of Parikh automata have received significant
attention over many
decades~\cite{ibarra1978reversal,greibach1978remarks,FOSSACS23,DBLP:journals/iandc/JantzenK03,DBLP:journals/ita/CadilhacFM12,DBLP:conf/icalp/BostanCKN20,DBLP:journals/ijfcs/CadilhacFM13,DBLP:conf/mfcs/FinkelS08,DBLP:conf/frocos/BersaniD11,DBLP:conf/icalp/Zetzsche13,DBLP:journals/mst/CadilhacKM18,DBLP:conf/icalp/Zetzsche16},
including a lot of work in recent
years~\cite{DBLP:conf/concur/ErlichGJL023,DBLP:conf/fsttcs/GuhaJL022,DBLP:conf/fsttcs/FiliotGM19,DBLP:conf/csl/GroblerSS24,DBLP:conf/mfcs/CadilhacG0R23,DBLP:journals/corr/abs-2505-13749}.
This is because they are expressive enough to model non-trivial counting
behavior, but still enjoy low complexity for many algorithmic tasks (e.g.\ the
emptiness problem is $\coNP$-complete).  Example applications are monadic
second-order logic with cardinalities~\cite{DBLP:conf/icalp/KlaedtkeR03} (this
paper introduced the specific model of Parikh automata), solving subword
constraints~\cite{DBLP:conf/lics/HalfonSZ17}, and model-checking FIFO channel systems~\cite{DBLP:journals/tcs/BouajjaniH99}. Moreover, these languages have
other equivalent characterizations, such as reversal-bounded counter
automata---a classic (and intensely studied) type of infinite-state systems
with nice decidability
properties~\cite{ibarra1978reversal,DBLP:conf/frocos/BersaniD11}---and automata
with $\Z$-counters, also called
$\Z$-VASS~\cite{greibach1978remarks,DBLP:conf/rp/HaaseH14}\footnote{See
\cite{FOSSACS23} for efficient translation among Parikh automata,
reversal-bounded counter automata, and $\Z$-VASS.}. 

Decidability of regular separability was shown by Clemente, Czerwi\'{n}ski,
Lasota, and Paperman~\cite{RegSepParikh} in 2017 as one of the first
decidability results for regular separability. Moreover, this result was a key ingredient in Keskin and
Meyer's algorithm to decide regular separability for general
VASS~\cite{DBLP:conf/lics/Keskin024}. However, despite the strong interest in
Parikh automata and in regular separability, the complexity of this problem
remained unknown. In \cite[Section 7]{RegSepParikh}, the authors provide an
elementary complexity upper bound.

\subparagraph{Separating semilinear sets: Monadic interpolants}
One of the steps in the algorithm from \cite{RegSepParikh} is to decide separability of sets defined in Presburger arithmetic, the first-order theory of $(\N;+,\le,0,1)$.
Separators of logically defined sets can also be viewed as \emph{interpolants}. If $\varphi(\bm{x},\bm{y})$ and $\psi(\bm{y},\bm{z})$ are (first-order or propositional) formulas such that $\forall \bm{x} \forall \bm{y} \forall \bm{z} \, (\varphi(\bm{x},\bm{y}) \to \psi(\bm{y},\bm{z}))$ holds, then a formula $\chi(\bm{y})$ is a \emph{Craig interpolant} if $\forall \bm{x} \forall \bm{y} \, (\varphi(\bm{x},\bm{y}) \to \chi(\bm{y}))$ and $\forall \bm{y} \forall \bm{z} \, (\chi(\bm{y}) \to \psi(\bm{y},\bm{z}))$ both hold. Here, $\bm{x},\bm{y},\bm{z}$ are each a vector of variables, meaning $\chi$ only mentions variables that
occur both in $\varphi$ and $\psi$. Equivalently, the set defined by $\chi$ is
a separator of the sets defined by the existential formulas $\exists
\bm{x}\colon\varphi(\bm{x},\bm{y})$ and
$\exists\bm{z}\colon\neg\psi(\bm{y},\bm{z})$. In Interpolation-Based Model
Checking (ITP)~\cite{DBLP:conf/cav/McMillan03,DBLP:journals/pieee/VizelWM15},
Craig interpolants are used to safely overapproximate sets of states: If
$\varphi$ describes reachable states and $\psi$ describes the set of safe
states, then $\chi$ overapproximates $\varphi$ without adding unsafe states.
Note that in Presburger logic there are implications that do not have a Craig
interpolant (this is in contrast to propositional logic). So, before constructing
an interpolant, a first step of ITP is to decide whether there even exists
such an interpolant.

In the case of Presburger arithmetic, the definable sets are the
\emph{semilinear sets}. For many infinite-state systems, the step relation (or
even the reachability relation) is semilinear, and thus, separators can play
the role of Craig interpolants in infinite-state model checking. For the
separators, a natural choice is the class of \emph{recognizable sets}, which
are those defined by \emph{monadic} Presburger formulas, meaning each atom
refers to at most one variable.  Monadic formulas have recently received
attention~\cite{DBLP:journals/jacm/VeanesBNB17,DBLP:journals/pacmpl/BergstrasserGLZ24,DBLP:conf/cade/HagueLRW20,DBLP:conf/icalp/HaaseKMMZ24}
because of their applications in query optimization in constraint
databases~\cite{DBLP:conf/gis/GrumbachRS98,kuper2013constraint} and symbolic
automata~\cite{DBLP:journals/jacm/VeanesBNB17}. 
Thus, deciding recognizable separability of semilinear sets can be viewed as
synthesizing monadic Craig interpolants. 

Recognizable separability was shown decidable by Choffrut and
Grigorieff~\cite{ChoG06} (see~\cite{SepReach} for an extension beyond
semilinear sets). This was a key ingredient for separability of Parikh
automata in \cite{RegSepParikh}. Choffrut and Grigorieff's algorithm has
elementary complexity~\cite[Section 7]{RegSepParikh}, but the exact complexity
of recognizable separability of semilinear sets remained unknown.

\subparagraph{Contribution}
Our \emph{first main result} is that for given existential Presburger formulas,
recognizable separability (i.e.\ monadic separability) is $\coNP$-complete.  In
particular, re\-cog\-niz\-able separability is $\coNP$-complete for given semilinear
representations. Moreover, our result implies that recognizable separability is
$\coNP$-complete for rational subsets of monoids $\Sigma^*\times\N^d$ as
considered by Choffrut and Grigorieff~\cite{ChoG06}.  Building on the methods
of the first result, our \emph{second main result} is that regular separability
for Parikh automata is $\coNP$-complete.

\subparagraph{Application I: Monadic decomposability} Our first main result strengthens a
recent result on monadic decomposability. A formula in Presburger arithmetic is
\emph{monadically decomposable} if it has a monadic equivalent.  It was shown
recently that (i)~deciding whether a given quantifier-free formula is
\emph{monadically decomposable} (i.e.\ whether it has a monadic equivalent) is
$\coNP$-complete~\cite[Theorem 1]{DBLP:conf/cade/HagueLRW20} (see 
\cite[Corollary 8.1]{DBLP:journals/pacmpl/BergstrasserGLZ24} for an alternative proof; and see \cite[Proposition 3]{DBLP:conf/fossacs/ChistikovHM22} for improved bounds for the approach in \cite{DBLP:conf/cade/HagueLRW20}), whereas (ii)~for
existential formulas, the problem is $\coNEXP$-complete~\cite[Corollary 3.6]{DBLP:conf/icalp/HaaseKMMZ24}.  Our first main
result strengthens (i): If $\varphi(\bm{x})$ is a quantifier-free formula, then
the sets defined by $\varphi(\bm{x})$ and $\neg\varphi(\bx)$ are separable by a
monadic formula if and only if $\varphi(\bx)$ is monadically decomposable.
Perhaps surprisingly, our $\coNP$ upper bound still holds for existential
Presburger formulas, for which monadic decomposability is known to be 
$\coNEXP$-complete\footnote{This is not a contradiction to the above reduction
from monadic decomposability to recognizable separation, since this reduction
would require complementing an existential formula.}.

\subparagraph{Application II: Regularity of Parikh automata} Another
consequence of our results is that regularity of deterministic Parikh automata,
i.e.\ deciding whether a given deterministic Parikh automaton accepts a regular
language, is $\coNP$-complete: Given a deterministic Parikh
automaton for a language $L\subseteq\Sigma^*$, one can construct in polynomial
time a Parikh automaton for $K=\Sigma^*\setminus L$.  Then, $L$ is regular if
and only if $L$ and $K$ are regularly separable.
Here, we assume that the semilinear target set
is given as a quantifier-free Presburger formula. Decidability of this problem
has been shown by Cadilhac, Finkel, and McKenzie~\cite[Theorem
25]{DBLP:journals/ijfcs/CadilhacFM13} (even in the more general case of
unambiguous constrained automata).

\subparagraph{Key ingredients} 
The existing elementary-complexity algorithm for recognizable separability of
semilinear sets works with semilinear representations and distinguishes two
cases: If in one component $j$, one of the input sets $S_1,S_2\subseteq\N^d$ is
bounded by some $b\ge 0$, then it considers each $x\in[0,b]$ and recursively
decides separability of $S_1[j\mapsto x]$ and $S_2[j\mapsto x]$, where $S_i[j\mapsto x]$ is just $S_i$
restricted to having $x$ in this bounded component. If, however, all components in both
sets are unbounded, then it checks feasibility of a system of linear
Diophantine equations. This approach leads to repeated intersection of
semilinear sets, and thus exponential time. We provide a characterization
(\cref{characterization-recognizable-separability}) that describes
inseparability directly as the non-empty intersection of two semilinear sets
$\hat{S}_1,\hat{S}_2\subseteq\N^d$ associated with $S_1,S_2$.  This easily
yields an $\NP$ procedure for inseparability, even if the input sets are given
as existential Presburger formulas.

This characterization is then the first key ingredient for deciding regular
separability of Parikh automata in $\coNP$. This is because in
\cite{RegSepParikh}, it is shown that, after some preprocessing, the languages
of Parikh automata $\cA_1$ and $\cA_2$ are separable if and only if two
semilinear sets $C_1,C_2\subseteq\N^d$ associated with $\cA_1$ and $\cA_2$ are
separable by a recognizable set. These semilinear sets consist of vectors, each
of which counts for some run of $\cA_i$, how many times each simple cycles
occurs in this run. Thus, our first result tells us that it suffices to decide whether $\hat{C}_1$ and $\hat{C}_2$ are disjoint. Unfortunately, the vectors of $C_1,C_2$ have exponential dimension
$d$, since there are exponentially many simple cycles in each $\cA_i$.  Thus,
applying our first result directly using existential Presburger arithmetic would only yield a $\coNEXP$ upper bound.

To avoid this blowup, the second key idea is to \emph{encode the vectors in $\hat{C}_1$ and $\hat{C}_2$ as words}, where the cycle occurrences appear as a concatenation in some order. By constructing $\Z$-VASS $\cW_1,\cW_2$ for the encodings of the vectors in $\hat{C}_1,\hat{C}_2$, we reduce separability to intersection emptiness of $\cW_1$ and $\cW_2$. The latter, in turn, easily reduces to non-reachability in a product $\Z$-VASS, which is in $\coNP$.

\section{Preliminaries}
By $\N=\{0,1,2,\ldots\}$ we denote the set of all non-negative integers.
Let $d\in\N$ be a number and $I\subseteq[1,d]$ be a set of indices. By $\pi_I\colon\N^d\to\N^I$ we denote the \emph{projection} of vectors in $\N^d$ to vectors in $\N^I$, i.e., $\restrict{\bv}{I}[i]=\bv[i]$ for each $\bv\in\N^d$ and $i\in I$. The \emph{support} of a vector $\bv\in\N^d$ is the set of all coordinates in $\bv$ with non-zero value, i.e.\ $\supp(\bv)=\{i\in[1,d]\mid\bv[i]\neq0\}$.

\subparagraph{Semilinear sets}
A set $S\subseteq\N^d$ is \emph{linear} if there is a vector
$\bu\in\N^d$ and a finite set $P\subseteq\N^d$ of so-called \emph{periods} such
that $S=\bu+P^*$ holds. Here, for $P=\{\bu_1,\ldots,\bu_n\}$, the set $P^*$ is defined as
$ P^*=\{\lambda_1\bu_1+\cdots+\lambda_n\bu_n\mid \lambda_1,\ldots,\lambda_n\in\N\}. $
A subset $S\subseteq\N^d$ is called \emph{semilinear} if it is a finite union
of linear sets. In case we specify $S$ by way of a finite union of linear sets, then we call this description a \emph{semilinear representation}. The set $S\subseteq\N^d$ is called \emph{hyperlinear} if there are finite sets
$B,P\subseteq\N^d$ such that $S=B+P^*$ holds.  It is well known that the
semilinear sets are precisely those definable in \emph{Presburger
arithmetic}~\cite{GinS66a}, the first-order theory of the
structure $(\N;+,\le,0,1,(\equiv_m)_{m\in\N\setminus\{0\}})$. Here $\equiv_m$ is the
predicate where $x\equiv_m y$ if and only if $x-y$ is divisible by $m$. By
quantifier elimination, every formula in Presburger
arithmetic has a quantifier-free equivalent.

\subparagraph{Parikh automata} Intuitively, a Parikh automaton has finitely
many control states and access to $d\ge 0$ counters. Upon reading a letter (or
the empty word), it can add a vector $\bu\in\N^d$ to its counters. Moreover,
for each state $q\in Q$, it specifies a target set $C_q\subseteq\N^d$. An input
word is accepted if at the end of the run, the accumulated counter values
belong to $C_q$, where $q$ is the state at the end of the run. Formally, a
\emph{Parikh automaton} is a tuple $\cA=(Q,\Sigma,T,q_0,(C_q)_{q\in Q})$, where
$Q$ is a finite set of states, $T\subseteq
Q\times(\Sigma\cup\{\varepsilon\})\times \N^d\times Q$ is its finite set of
\emph{transitions}, $q_0\in Q$ is the \emph{initial} state, and $C_q\subseteq\N^d$
is the \emph{target set} in state
$q$, for each $q\in Q$. For an input word $w\in\Sigma^*$, a \emph{run on $w$} is
a sequence $(q_0,w_1,\bu_1,q_1)\cdots (q_{n-1},w_n,\bu_n,q_n)$ of transitions
in $T$ with $w=w_1\cdots w_n$. The run is \emph{accepting} if $\bu_1+\cdots+\bu_n\in C_{q_n}$. The
\emph{language} of $\cA$ is then the set of all words $w\in\Sigma^*$ such that
$\cA$ has an accepting run on $w$. 

\begin{remark}
	For our results on general Parikh automata, we assume that the target sets are
	specified using existential Presburger formulas. However, this is not
	an important aspect: Given a Parikh automaton, one can in polynomial
	time modify the automaton (and the target set) so that the target set
	is given, e.g.\ by a semilinear representation, or a quantifier-free
	Presburger formula. This is a simple consequence of the fact that one
	can translate Parikh automata into integer VASS in logarithmic
	space~\cite[Corollary 1]{FOSSACS23}.
	However, this conversion does not preserve determinism, and for
	deterministic Parikh automata, it can be important how target sets are
	given (see \cref{main-regularity} and the discussion after it).
	Therefore, for deterministic Parikh automata, we always specify
	how the targets sets are given.
\end{remark}

\subparagraph{Separability}
A subset $L\subseteq M$ of a monoid $M$ is \emph{recognizable} if there is a
morphism $\varphi\colon M\to F$ into some finite monoid $F$ such that
$\varphi^{-1}(\varphi(L))=L$.  The recognizable subsets of $M$ form a Boolean
algebra~\cite[Chapter III, Prop.~1.1]{Berstel}. We say that sets $K,L\subseteq
M$ are \emph{(recognizably) separable}, denoted $\regsep{K}{L}$, if there is a
morphism $\varphi\colon M\to F$ into some finite monoid $F$ such that
$\varphi(K)\cap\varphi(L)=\emptyset$. Equivalently, we have $\regsep{K}{L}$ if and
only if there is a recognizable $S\subseteq M$ with $K\subseteq S$ and
$S\cap L=\emptyset$.  Here, $S$ is called a \emph{separator} of $K$ and $L$. Clearly, we have $\regsep{K}{L}$ if and only if $\regsep{L}{K}$: if $S$ is a separator of $K$ and $L$ then $M\setminus S$ separates $L$ and $K$.

In the case $M=\Sigma^*$ for some alphabet $\Sigma$, the recognizable sets
in $\Sigma^*$ are exactly the regular languages (cf.\
\cite[Theorem~II.2.1]{Sak09}), and thus we speak of \emph{regular
separability}. In the case $M=\N^d$ for some $d\ge 0$, then the recognizable
subsets of $\N^d$ are precisely the finite unions of cartesian products
$U_1\times\cdots\times U_d$, where each $U_i\subseteq\N$ is ultimately
periodic~\cite[Theorem 5.1]{Berstel}. Here, a set $U\subseteq\N$ is
\emph{ultimately periodic} if there are $n_0,p\in\N\setminus\{0\}$ such that for all $n\ge
n_0$, we have $n\in U$ if and only if $n+p\in U$. This implies that the
recognizable subsets of $\N^d$ are precisely those definable by a \emph{monadic
Presburger formula}, i.e.\ one where every atom only refers to one
variable~\cite{DBLP:journals/jacm/VeanesBNB17}.  For these reasons, in the case
of $M=\N^d$, we also sometimes speak of \emph{monadic separability}.

In a \emph{recognizable separability problem}, we are given two subsets $K$ and $L$ from a monoid $M$ as input, and we want to decide whether $K$ and $L$ are recognizably separable. Again, in the case of $M=\Sigma^*$, we also call this the \emph{regular separability problem}.

\section{Main results}
\subparagraph{Recognizable separability of semilinear sets}
Our first main result is the following.
\begin{theorem}\label{main-epa}
Given two semilinear sets defined by existential Presburger formulas, recognizable separability is $\coNP$-complete.
\end{theorem}
The lower bound follows with a simple reduction from the emptiness problem for sets defined by existential Presburger formulas: If $\varphi$ defines a subset $K\subseteq\N^d$, then $\regsep{K}{\N^d}$ if and only if $K$ is empty. We prove the $\coNP$ upper bound in \cref{section-epa}. By the same argument, recognizable separability is $\coNP$-hard for input sets given by quantifier-free formulas. Thus:
\begin{corollary}\label{main-quantifier-free}
Given two semilinear sets defined by quantifier-free Presburger formulas, recognizable separability is $\coNP$-complete.
\end{corollary}
In particular, this re-proves the $\coNP$ upper bound for monadic decomposability of quantifier-free formulas, as originally shown by Hague, Lin, R\"{u}mmer, and Wu~\cite[Theorem~1]{DBLP:conf/cade/HagueLRW20}. 
\begin{remark}
	Our result also implies that for existential Presburger formulas over
	$(\Z;+,\le,0,1,(\equiv_m)_{m\in\N\setminus\{0\}})$ defining $K,L\subseteq\Z^d$, it is $\coNP$-complete to
	decide whether they are separable by a monadically defined subset of
	$\Z^d$. Indeed, consider the injective map $\nu\colon\Z^d\to\N^{2d}$, where $\nu(x_1,\ldots,x_d)=(\sigma(x_1),|x_1|,\ldots,\sigma(x_d),|x_d|)$ with $\sigma(x)=0$ for $x\ge 0$ and $\sigma(x)=1$ for $x<0$. Then $S\subseteq\Z^d$ is monadically definable if and only if $\nu(S)$ is monadically definable\footnote{This is easily shown by translating each atomic formula (over a single variable) into a monadic formula in each direction. However, note that within $\Z^d$, monadic definability is not the same as recognizability. For example, the sets $\{0\}$ and $\Z\setminus\{0\}$ are monadically separable, but not separable by a recognizable subset of $\Z$, since every non-empty recognizable subset of $\Z$ is infinite~\cite[Chapter III, Example 1.4]{Berstel}.}. Thus, $K,L\subseteq\Z^d$ are monadically separable if and only if $\nu(K),\nu(L)\subseteq\N^{2d}$ are monadically separable. Finally, one easily constructs existential formulas for $\nu(K),\nu(L)$.
\end{remark}

Since for a given semilinear representation of a set $S\subseteq\N^d$, it is
easy to construct an existential Presburger formula defining $S$,
\cref{main-epa} also implies the following.
\begin{corollary}\label{main-semilinear}
Given two semilinear representations, recognizable separability is $\coNP$-complete.
\end{corollary}
In this case, the $\coNP$ lower bound comes from the $\NP$-hard membership
problem for semilinear sets (even if all numbers are written in unary) \cite[Lemma~10]{IbaR14}: For a
semilinear subset $S\subseteq\N^d$ and a vector $\bu\in\N^d$, we have
$\bu\notin S$ if and only if $\regsep{S}{\{\bu\}}$.
Finally, \cref{main-epa} allows us to settle the complexity of
recognizable separability of rational subsets of $\Sigma^*\times\N^d$.
\begin{corollary}
Given $d\in\N$ and two rational subsets of $\Sigma^*\times\N^d$, deciding recognizable separability is $\coNP$-complete.
\end{corollary}
Decidability was first shown by Choffrut and Grigorieff~\cite[Theorem
1]{ChoG06}. The $\coNP$ upper bound follows because Choffrut and
Grigorieff~\cite[Theorem 10]{ChoG06} reduce recognizable separability of
subsets of $\Sigma^*\times\N^d$ to recognizable separability of rational subsets of
$\N^{2d}$ (and their reduction is clearly in polynomial time). Moreover, for a
given rational subset of $\N^{2d}$, one can construct in polynomial time an
equivalent existential Presburger formula~\cite[Theorem
1]{DBLP:conf/icalp/SeidlSMH04}. Thus, the upper bound follows from
\cref{main-epa}. Since semilinear sets in $\N^d$ (given by a semilinear representation) can be viewed as rational subsets of $\N^d$ (and hence of $\Sigma^*\times\N^d$), the $\coNP$ lower bound is inherited from \cref{main-semilinear}.

\subparagraph{Regular separability of Parikh automata} Our second main result is the following:
\begin{theorem}\label{main-parikh}
Regular separability for Parikh automata is $\coNP$-complete.
\end{theorem}
The $\coNP$ lower bound comes via the $\coNP$-complete emptiness problem:
For a given Parikh automaton accepting a language $K\subseteq\Sigma^*$, we have
$\regsep{K}{\Sigma^*}$ if and only if $K=\emptyset$. Thus, the interesting part is
the upper bound, which we prove in \cref{section-parikh-automata}. This is a significant
improvement to the previously known elementary (or finitely iterated exponential time)
complexity upper bound by Clemente, Czerwi\'{n}ski, Lasota, and Paperman~\cite{RegSepParikh}.

\Cref{main-parikh} can also be applied to deciding regularity of deterministic Parikh automata.
\begin{corollary}\label{main-regularity}
For deterministic Parikh automata with target sets given as quantifier-free Presburger formulas, deciding regularity is $\coNP$-complete.
\end{corollary}
Decidability of regularity was shown by Cadilhac, Finkel, and McKenzie~\cite[Theorem
25]{DBLP:journals/ijfcs/CadilhacFM13} (in the slightly more general setting of unambiguous constrained automata).
For the $\coNP$ upper bound, note that for a language $L\subseteq\Sigma^*$ given by a deterministic Parikh automaton (with quantifier-free formulas for the target sets), one can in polynomial time construct the same type of automaton for the complement $\Sigma^*\setminus L$. Since $L$ is regular if and only if $L$ and $\Sigma^*\setminus L$ are separable by a regular language, we can invoke \cref{main-parikh}.
The $\coNP$ lower bound is inherited from monadic decomposability of quantifier-free formulas. 
Indeed, given a quantifier-free Presburger formula $\varphi(x_1,\ldots,x_n)$ with free variables $(x_1,\ldots,x_n)$, one easily constructs a deterministic Parikh automaton (with quantifier-free target sets) for the language $L_\varphi=\{a_1^{x_1}\cdots a_n^{x_n} \mid \varphi(x_1,\ldots,x_n)\}$. As shown by Ginsburg and Spanier~\cite[Theorem 1.2]{ginsburg1966bounded}, $L_\varphi$ is regular if and only if $\varphi$ is monadically decomposable. However, monadic decomposability for quantifier-free formulas is $\coNP$-complete~\cite[Theorem 1]{DBLP:conf/cade/HagueLRW20}.

For the $\coNP$ upper bound in \cref{main-regularity}, we cannot drop the
assumption that the formula be quantifier-free. This is because if the target
sets can be existential Presburger formulas, then the regularity problem is
$\coNEXP$-hard. This follows by the same reduction from monadic
decomposability: If we construct $L_\varphi$ as above using an existential
formula $\varphi$, then again, $L_\varphi$ is regular if and only if $\varphi$
is monadically decomposable. Moreover, monadic decomposability for existential
formulas is $\coNEXP$-complete~\cite[Corollary
3.6]{DBLP:conf/icalp/HaaseKMMZ24}.

\section{A characterization of separability in hyperlinear sets}\label{sec:hyperlinear}

Before we prove our two main results, \cref{main-epa,main-parikh}, we should
recall the ideas of the existing algorithms~\cite{ChoG06,SepReach} for
recognizable separability of linear sets. We will use these ideas to
obtain a new characterization of separability in hyperlinear sets.

Let $L_1,L_2\subseteq\N^d$ be two linear sets. The
algorithms~\cite{ChoG06,SepReach} rely on a procedure that successively
eliminates ``bounded components'': If, say, $L_1$ is bounded in component $j$
by some $b\in\N$, then one can observe that $\regsep{L_1}{L_2}$ if, and only if,
$\regsep{L_1[j\mapsto x]}{L_2[j\mapsto x]}$ for every $x\in[0,b]$. Here,
$L_i[j\mapsto x]$ is $L_i$ restricted to those vectors that have $x$ in the
$j$-th component, and then projected to all components $\ne j$. Therefore,
the algorithms of \cite{ChoG06,SepReach} recursively check separability of
$L_1[j\mapsto x]$ and $L_2[j\mapsto x]$ for each $x\in[0,b]$.
This process invokes several expensive intersection operations on semilinear
sets and thus has high complexity.  Instead, our approach immediately guesses
and verifies the set of components that remain after the elimination process.
The corresponding checks involve the notion of twin-unboundedness.

\subparagraph{Twin-unbounded components} 
Our notion applies, slightly more generally, to hyperlinear sets. Hence, let
$R=A+U^*\subseteq\N^d$ and $S=B+V^*\subseteq\N^d$ be two hyperlinear sets where
$A,B,U,V\subseteq\N^d$ are finite sets. 

\begin{definition}
  A coordinate $j\in[1,d]$ is \emph{twin-unbounded for $R$ and $S$} if there
  exist $\bp\in U^*$ and $\bq\in V^*$ such that $j\in\supp(\bp)=\supp(\bq)$.
\end{definition}
Hence, intuitively, twin-unbounded coordinates are those that can be made large/driven up in $R$ in the same way as in $S$.
There is yet another characterization of twin-unbounded coordinates. Let $j \in [1,d]$. We say the $j$-th coordinate of the hyperlinear set $S=B+V^*$ is \emph{bounded} if there is no period vector in $V$ with support on $j$, i.e., $j\notin\supp(\bp)$ for all $\bp\in V$. We say that a subset $J\subseteq[1,d]$ of coordinates is bounded in $S$ if each $j\in J$ is bounded in $S$.
Consider the following process: Given two hyperlinear sets $R$ and $S$. We modify $R$ and $S$ by performing each of the following three steps for each coordinate $j\in[1,d]$ until the sets of remaining period vectors in $R$ and $S$ stabilize:
\begin{itemize}
  \item If neither $R$ nor $S$ is bounded at $j$, we leave $S$ and $R$ untouched.
  \item If only $R$ is bounded at $j$, we remove all period vectors from $S$ which have support on $j$.
  \item If only $S$ is bounded at $j$, we remove all period vectors from $R$ which have support on $j$.
\end{itemize}
Then, the coordinates that remain unbounded are precisely the twin-unbounded ones.

\begin{example}
  Consider $R=\{(1,0,1)\}^*$ and $S=\{(1,1,0),(0,0,1)\}^*$. Then $R$ is bounded by the value $0$ at coordinate $2$. So $R$ and $S$ are separable if and only if $R$ and $S$ restricted to the vectors having the value $0$ in the second coordinate. So, we only consider this restriction of $S$---in our algorithm this is reflected by the deletion of the period vector $(1,1,0)$ of $S$. After deletion of the period vector $(1,1,0)$, $S$ is bounded at coordinate $1$ by the value $0$. So, we remove the period vector $(1,0,1)$ from $R$. Finally, the period vector $(0,0,1)$ of $S$ gets removed since $R$ is now bounded at coordinate $3$. Hence, our algorithm terminates in this case with no twin-unbounded coordinates. This example shows that even if $R$ and $S$ both are unbounded in coordinates $1$ and $3$, none of these coordinates is twin-unbounded.
  
  If $R=\{(1,0,1),(0,1,0)\}^*$ and $S=\{(1,1,0),(0,0,1)\}^*$, then no coordinate is bounded in $R$ and $S$. Hence, all coordinates are twin-unbounded and no period vector gets removed.
\end{example}

For $J\subseteq[1,d]$, we write $U_J=\{\bp\in U\mid
\supp(\bp)\subseteq J\}$ and $V_J=\{\bq\in V\mid \supp(\bq)\subseteq J\}$.

\subparagraph{Separating by modular constraints}
As observed in \cite{ChoG06,SepReach}, if all coordinates of two linear sets $L_1,L_2$ are unbounded, then separability holds if and only if the two sets can be separated by modulo constraints. This relies on the well known fact that finitely generated abelian groups are \emph{subgroup separable}, i.e.\ that for every element $\bu\in\Z^d$ that does not belong to a subgroup $A\subseteq\Z^d$, there exists a homomorphism $\varphi\colon\Z^d\to \bbF$ into a finite group $\bbF$ such that (i)~$A$ is included in the kernel of $\varphi$ and (ii)~$\varphi(\bu)\ne 0$. In our characterization (\cref{characterization-recognizable-separability}) we will use similar arguments and therefore we will recall subgroup separability here. We include a short proof in\ %
\ifFullVersion%
  \cref{app:hyperlinear}.%
\else%
  the full version.%
\fi%
\begin{restatable}[Subgroup separability]{lemma}{subgroupSeparability}\label{subgroup-separability}
If $A\subseteq\Z^d$ is a subgroup and $\bu\in \Z^d\setminus A$, then there
	is an $s\in\N$, $s>0$, and a morphism $\varphi\colon\Z^d\to\Z/s\Z$ with (i)~$\varphi(A)=0$ and (ii)~$\varphi(\bu)\ne 0$.
\end{restatable}

\subparagraph{Separability vs.\ intersection emptiness}
We will now characterize inseparability of hyperlinear sets $R,S$ via the intersection of two hyperlinear sets $\hat{R}$ and $\hat{S}$ associated with $R,S$. The proof will rely on an equivalence relation of vectors.
For vectors $\bu,\bv\in\N^d$ and $k\in\N\setminus\{0\}$, we write $\bu\sim_k\bv$ if for every $i\in[1,d]$, we have 
\begin{enumerate}[(1)]
\item $\bu[i]=\bv[i]\le k$ or
\item $\bu[i],\bv[i]>k$ and $\bu[i]\equiv \bv[i]\bmod{k}$.
\end{enumerate}

The following was shown in \cite[Prop. 18]{SepReach}.
\begin{lemma}\label{lem:insepSim}
  For any sets $X,Y\subseteq\N^d$, the following are equivalent:
  \begin{enumerate}[(1)]
    \item $X$ and $Y$ are not separable by a recognizable set.
    \item for each $k\in\N\setminus\{0\}$ there are $\bx_k\in X$ and $\by_k\in Y$ with $\bx_k\sim_k \by_k$.
  \end{enumerate}
\end{lemma}
Let $k,\ell\in\N\setminus\{0\}$ be such that $k$ divides $\ell$. We can observe that $\bu\sim_\ell\bv$ implies $\bu\sim_k\bv$ in this case. Thus, to show recognizable inseparability of two sets $X,Y\subseteq\N^d$, it suffices to find $\bx_k\in X$ and $\by_k\in Y$ for almost all numbers $k\in\N\setminus\{0\}$. We will use this fact in the proof of the following characterization of inseparability.

\begin{proposition}\label{characterization-recognizable-separability}
Let $R=A+U^*\subseteq\N^d$ and $S=B+V^*\subseteq\N^d$ be hyperlinear sets.  Then
$R$ and $S$ are \emph{not} separable by a recognizable set if and only if  the intersection
\begin{equation} (A+U^*-U_J^*)\cap (B+V^*-V_J^*)\label{inseparability-intersection}\end{equation}
is non-empty, where $J\subseteq[1,d]$ is the set of coordinates that are twin-unbounded for $R,S$.
\end{proposition}
\begin{proof}
Suppose there is a vector $\bx$ in the intersection \eqref{inseparability-intersection}. Then we can write $\bx=\bu-\bar{\bu}$ and $\bx=\bv-\bar{\bv}$ with $\bu\in A+U^*$, $\bv\in B+V^*$, $\bar{\bu}\in U_J^*$, and $\bar{\bv}\in V_J^*$.
Since $J$ is twin-unbounded for $R$ and $S$, there are---by definition---$\bp_j\in U^*$ and $\bq_j\in V^*$ with $j\in\supp(\bp_j)=\supp(\bq_j)$ for each $j\in J$. Then for $\bp:=\sum_{j\in J}\bp_j$ and $\bq:=\sum_{j\in J}\bq_j$ we infer $J\subseteq\supp(\bp)=\supp(\bq)$. Now for each $k\in\N\setminus\{0\}$, consider the vectors
\begin{equation*}
  \bu_k=\bu-\bar{\bu}+2k\cdot \bp+k\cdot\bar{\bu}\quad\text{and}\quad
  \bv_k=\bv-\bar{\bv}+2k\cdot \bq+k\cdot\bar{\bv}\,.
\end{equation*}
Then we have $\bu_k,\bv_k\in\N^d$ for each $k\in\N\setminus\{0\}$. We claim that $\bu_k\sim_k\bv_k$ for all $k$. Indeed, on coordinates $j\in[1,d]\setminus \supp(\bp)$, the vectors $\bu_k$ and $\bv_k$ coincide with $\bx$. Moreover, on coordinates $j\in \supp(\bp)$, both vectors $\bu_k$ and $\bv_k$ are larger than $k$ and also congruent to $\bx[j]\bmod{k}$. Hence, $\bu_k\sim_k\bv_k$. Since clearly $\bu_k=\bu+2k\cdot\bp+(k-1)\cdot\bar{\bu}\in R$ and $\bv_k=\bv+2k\cdot\bq+(k-1)\cdot\bar{\bv}\in S$, \cref{lem:insepSim} implies that $R$ and $S$ are not separable.

Conversely, suppose that $R$ and $S$ are not separable. Then by \cref{lem:insepSim} there are $\bu_k\in R$ and $\bv_k\in S$ with
$\bu_k\sim_k\bv_k$ for every $k\in\N\setminus\{0\}$. 
	We claim that the sequences $\bu_1,\bu_2,\ldots$ and $\bv_1,\bv_2,\ldots$ have subsequences $\bu'_1,\bu'_2,\ldots$ and $\bv'_1,\bv'_2,\ldots$ such that for every $k\ge 1$, we have (i)~$\bu'_{k+1}\in \bu'_k+U_J^*$, (ii)~$\bv'_{k+1}\in \bv'_k+V_J^*$ and (iii)~$\bu'_k\sim_k\bv'_k$.

	The claim is easy to observe: Note that by picking subsequences, we may assume that even $\bu_k\sim_{k!}\bv_k$ for every $k\ge 1$. Moreover, the latter property is preserved by taking subsequences. Thus, since $A,B$ are finite, by picking subsequences again, we may assume that there are $\br\in A$ and $\bs\in B$ such that $\bu_k\in \br+U^*$ and $\bu_k\in\bs+V^*$ and $\bu_k\sim_{k!}\bv_k$ for $k\ge 1$. Then, by Dickson's lemma, we may assume that in addition $\bu_{k+1}\in \bu_k+U^*$ and $\bv_{k+1}\in\bv_k+V^*$ for every $k\ge 1$ (here, we apply Dickson's lemma to the $|U|$-dimensional vectors of coefficients at period vectors in $U$ and similarly for $V$).
Now since $\bu_k\sim_{k!}\bv_k$ for every $k$, it follows that the
sequences $\bu_1,\bu_2,\ldots$ and $\bv_1,\bv_2,\ldots$ are unbounded on the
same set $J\subseteq[1,d]$ of coordinates. Then clearly, $J$ is 
twin-unbounded for $R$ and $S$. This means, for all but finitely many $k$, we have 
$\bu_{k+1}\in \bu_k+U_J^*$ and $\bv_{k+1}\in\bv_k+V_J^*$. Hence, by picking another subsequence, we may assume that the latter holds for every $k\ge 1$. Then, $\bu_1,\bu_2,\ldots$ and $\bv_1,\bv_2,\ldots$ satisfy the properties (i--iii) above, establishing our claim.

We now claim that $\bu_1-\bv_1$ belongs to the group $\langle U_J\cup V_J\rangle$ generated by $U_J\cup V_J$.
	Towards a contradiction, suppose $\bu_1-\bv_1$ does not belong to $\langle U_J\cup V_J\rangle$. By \cref{subgroup-separability}, there must be an $s\in\N$, $s>0$, and a morphism $\varphi\colon\Z^d\to\Z/s\Z$ such that $\varphi(\langle U_J\cup V_J\rangle)=0$ and $\varphi(\bu_1-\bv_1)\ne 0$. However, the vector
	\[ (\bu_s-\bv_s)-(\bu_1-\bv_1) = \underbrace{(\bu_s-\bu_1)}_{\in \langle U_J\rangle}-\underbrace{(\bv_s-\bv_1)}_{\in \langle V_J\rangle}\]
	belongs to $\langle U_J\cup V_J\rangle$, but also agrees with $\bu_1-\bv_1$ under $\varphi$ (since all components of $\bu_s-\bv_s$ are divisible by $s$), contradicting \cref{subgroup-separability}. Hence $\bu_1-\bv_1\in\langle U_J\cup V_J\rangle$.

This means, we can write $\bu_1-\bv_1=\bv-\bar{\bv}-(\bu-\bar{\bu})$ with $\bu,\bar{\bu}\in U_J^*$ and $\bv,\bar{\bv}\in V_J^*$. But then the vector
$\bu_1+\bu-\bar{\bu}=\bv_1+\bv-\bar{\bv}$
belongs to the intersection \eqref{inseparability-intersection}.
\end{proof}

With \cref{characterization-recognizable-separability}, we have now
characterized inseparability of subsets of $\N^d$ via a particular intersection
of two sets in $\Z^d$. It will later be more convenient to work with
intersections of sets in $\N^d$, which motivates the following
reformulation of \cref{characterization-recognizable-separability}.
\begin{theorem}\label{characterization-recognizable-separability2}
  Let $R=A+U^*\subseteq\N^d$ and $S=B+V^*\subseteq\N^d$ be hyperlinear sets.  Then
  $R$ and $S$ are \emph{not} separable by a recognizable set if and only if  the intersection
  \begin{equation}
    (A+U^*+V_J^*)\cap (B+V^*+U_J^*)\label{inseparability-intersection2}
  \end{equation}
  is non-empty, where $J\subseteq[1,d]$ is the set of coordinates that are twin-unbounded for $R,S$.
\end{theorem}
\begin{proof}
	Direct consequence of \cref{characterization-recognizable-separability}, since clearly $A+U^*-U_J^*$ intersects $B+V^*-V_J^*$ if and only if $A+U^*+V_J^*$ intersects $B+V^*+U_J^*$.
\end{proof}

\section{Separability of semilinear sets is in $\coNP$}\label{section-epa}
Using the characterization \cref{characterization-recognizable-separability2},
we can now explain our algorithm for the $\coNP$ upper bound in \cref{main-epa}.
We describe an $\NP$ algorithm that establishes \emph{inseparability}.

\subparagraph{Algorithm Step I: Solution sets to linear Diophantine equations}
Let us first see that we can reduce the problem to the case where both input
sets are given as projections of solution sets of linear Diophantine
equations. We may assume that the input formulas are of the form
$\exists \bx\colon \kappa(\bx,\by)$, where $\kappa$ is
a formula consisting of conjunction and disjunction (i.e.\ no negation) of
atoms of the form $t\ge a$, where $t$ is a linear combination of variables
$\bx=(x_1,\ldots,x_n), \by=(y_1,\ldots,y_m)$ and integer coefficients, and $a$ is a constant. 

Let $\varphi$ be a formula as described above. It is a well known fact that
$\varphi$ can be transformed into disjunctive normal form. This means,
$\varphi$ is equivalent to a formula $\varphi_1\vee\cdots\vee\varphi_k$,
where each $\varphi_i$ (a so-called \emph{clause}) has the form
$\exists\bx\colon\xi(\bx,\by)$ such that $\xi$ is a conjunction of atoms
appearing in $\varphi$. In general, the number of clauses of $\varphi$
is exponential.

Now, let $\varphi$ and $\psi$ be the input formulas of the algorithm and
let $\varphi_1\vee\cdots\vee\varphi_k$ and $\psi_1\vee\cdots\vee\psi_\ell$
be their equivalent formulas in disjunctive normal form. Since the number
of clauses is exponential, we cannot compute all clauses for $\varphi$
and $\psi$.
However, the solution sets of $\varphi$ and $\psi$ are recognizably inseparable
if, and only if, for some pair $i,j$, the solution sets of the formulas
$\varphi_i$ and $\psi_j$ are recognizably inseparable.  This is 
due to the following fact, which follows standard ideas (see\ %
\ifFullVersion%
  \cref{app:epa}\ %
\else%
  the full version\ %
\fi%
for a proof in this particular setting).%

\begin{restatable}{lemma}{lemRepSepUnion}\label{lem:regSepUnion}
  Let $K,K_1,\ldots,K_n,L\subseteq M$ be sets from a monoid $M$ such
  that $K=K_1\cup\cdots\cup K_n$. Then $\regsep{K}{L}$ if, and only if,
  $\regsep{K_i}{L}$ for all $1\leq i\leq n$.
\end{restatable}

Thus, for deciding the inseparability of the solution sets of $\varphi$ and $\psi$
in $\NP$ it is sufficient to guess (in polynomial time) clauses $\varphi_i$
and $\psi_j$ and show that inseparability of the solution sets of these
two formulas is decidable in $\NP$. Therefore, from now on we can assume
that the input formulas are (existentially quantified) conjunctions of
atoms of the form $t\ge a$.

In particular, each of the two input sets is a projection of the solution set of a system of linear Diophantine inequalities. By introducing slack variables (which will also be projected away), we can turn \emph{inequalities} into \emph{equations}. Thus, we have as input sets $K,L\subseteq\N^d$ with
\begin{equation}\label{input-sets-inequalities}
  K=\pi(\{\bx\in\N^r \mid A\bx=\bb\})\quad\text{and}\quad
  L=\pi(\{\bx\in\N^r \mid C\bx=\bd \})\,,
\end{equation}
where $\pi\colon\Z^r\to\Z^d$ is the projection to the first $d$ components, and
$A,C\in\Z^{s\times r}$ are integer matrices, and $\bb,\bd\in\Z^s$ are integer
vectors. Note that here, assuming that the numbers $r$ of columns and the
number $s$ of rows is the same for $K$ and $L$ means no loss of generality.

\subparagraph{Algorithm Step II: Recognizable inseparability as satisfiability}
In the second step, we will reduce recognizable inseparability of $K$ and $L$ to
satisfiability of an existential Presburger formula. To this end, we use the
fact that the solution sets to $A\bx\ge\bb$ (resp.\ $C\bx\ge\bd$) are
hyperlinear sets, which allows us to apply
\cref{characterization-recognizable-separability2}.
\begin{restatable}{proposition}{inseparabilityViaInequalities}\label{inseparability-via-inequalities}
$K$ and $L$ are recognizably inseparable if, and only if, there are
vectors $\bp,\bq,\bu,\bv,\bx,\by\in\N^r$ with
\begin{enumerate}[(1)]
\item $A\bp=\bzero$, $C\bq=\bzero$, $\supp(\pi(\bp))=\supp(\pi(\bq))$, 
\item $\supp(\pi(\bu)),\supp(\pi(\bv))\subseteq \supp(\pi(\bp))$,
$A\bu=\bzero$, and $C\bv=\bzero$,
\item $A\bx=\bb$, $C\by=\bd$, and
$\pi(\bx+\bv)=\pi(\by+\bu)$.
\end{enumerate}
\end{restatable}
\begin{proof}
  We apply \cref{characterization-recognizable-separability2}. To this
  end, we use the standard hyperlinear representation for solution sets
  of systems of linear Diophantine equalities. Let $A_0\subseteq\N^r$
  be the set of all (component-wise) minimal solutions to $A\bx=\bb$, and let
  $U\subseteq\N^r$ be the set of all minimal solutions to
  $A\bx=\bzero$. Then it is well known that $A_0$ and $U$ are finite and also
  $K=\pi(A_0+U^*)=\pi(A_0)+\pi(U)^*$. In the same way, we obtain a
  hyperlinear representation $L=\pi(B_0+V^*)=\pi(B_0)+\pi(V)^*$. Then, we
  can show the \lcnamecref{inseparability-via-inequalities} using
  \cref{characterization-recognizable-separability2}. For a full proof, see\ %
  \ifFullVersion%
    \cref{app:epa}.%
  \else%
    the full version.%
  \fi%
\end{proof}

Finally, \cref{inseparability-via-inequalities} can be used to
complete the proof of our first main result:

\begin{proof}[Proof of \cref{main-epa}]
  Let $\varphi$ and $\psi$ be two existential Presburger formulas
  without negation and using only atoms of the form $t\geq0$, where
  $t$ is a linear combination of variables and integer coefficients.
  We give an $\NP$ algorithm deciding inseparability by a recognizable
  set.
  
  Since the solution sets of $\varphi$ and $\psi$ are inseparable if,
  and only if, their disjunctive normal forms have at least one pair
  of inseparable clauses, we guess such a pair of these clauses
  $\varphi_i$ and $\psi_j$ (cf.\ \cref{lem:regSepUnion}). We can
  transform $\varphi_i$ and $\psi_j$ into Diophantine equations
  $A\bx=\bb$ and $C\bx=\bd$. Using \cref{inseparability-via-inequalities}
  we obtain in polynomial time an existential Presburger formula that
  is satisfiable if, and only if, the solution sets of $A\bx=\bb$ and
  $C\bx=\bd$ are inseparable if, and only if, $\varphi_i$ and $\psi_j$
  are inseparable. Finally, the result follows from $\NP$-completeness of the  
  existential fragment of Presburger arithmetic.
\end{proof}

\section{Regular separability of Parikh automata}\label{section-parikh-automata}
We now prove our second main result: the $\coNP$ upper
bound of regular separability of Parikh automata (\cref{main-parikh}).
For this, it will be technically simpler to work with $\Z$-VASS, which are equivalent to Parikh automata. In \cite[Corollary 1]{FOSSACS23}, it was shown that the two automata models can be converted (while preserving the language) into each other in logarithmic space. Therefore, showing the $\coNP$ upper bound for $\Z$-VASS implies it for Parikh automata.

\subparagraph{Integer VASS}
A ($d$-dimensional) \emph{integer vector addition system with states} (\emph{$\Z$-VASS}, for short) is a quintuple $\cV=(Q,\Sigma,T,\iota,f)$ where $Q$ is a finite set of \emph{states}, $\Sigma$ is an \emph{alphabet}, $T\subseteq Q\times \Sigma_\varepsilon\times \Z^d\times Q$ is a finite set of \emph{transitions}, and $\iota,f\in Q$ are its \emph{source} and \emph{target state}, respectively. Here, $\Sigma_\varepsilon=\Sigma\cup\{\varepsilon\}$. A $\Z$-VASS $\cV=(Q,\Sigma,T,\iota,f)$ is called \emph{deterministic} if $\cV$ has no $\varepsilon$-labeled transitions and for each $p\in Q$ and $a\in\Sigma$ there is at most one transition of the form $(p,a,\bv,q)\in T$ (where $\bv\in\Z^d$ and $q\in Q$).

A \emph{configuration} of $\cV$ is a tuple from $Q\times\Z^d$. For two configurations $(p,\bu),(q,\bv)$ and a word $w\in\Sigma^*$ we write $(p,\bu)\xrightarrow{w}_\cV(q,\bv)$ if there are states $q_0,q_1,\ldots,q_\ell\in Q$, vectors $\bv_0,\bv_1,\ldots,\bv_\ell\in\Z^d$, and letters $a_1,\ldots,a_\ell\in\Sigma_\varepsilon$ such that $w=a_1a_2\cdots a_\ell$, $(p,\bu)=(q_0,\bv_0)$, $(q,\bv)=(q_\ell,\bv_\ell)$, and for each $1\leq i\leq\ell$ we have a transition $t_i=(q_{i-1},a_i,\bx_i,q_i)\in T$ with $\bv_i=\bv_{i-1}+\bx_i$. In this case, the sequence $t_1t_2\cdots t_\ell$ is called a ($w$-labeled) run of $\cV$. The \emph{accepted language} of $\cV$ is $L(\cV)=\{w\in\Sigma^*\mid(\iota,\bzero)\xrightarrow{w}_\cV(f,\bzero)\}$.

Let $I\subseteq[1,d]$ be a set of indices. Then we can generalize the acceptance behavior of the $\Z$-VASS $\cV$ as follows:
\[
  L(\cV,I)=\bigl\{w\in\Sigma^*\,\bigm|\,
    \exists\bv\in\Z^d\colon(\iota,\bzero)\xrightarrow{w}_\cV(f,\bv)
    \text{ and }\restrict{\bv}{I}=\bzero
  \bigr\}\,.
\]
Note that $L(\cV,[1,d])=L(\cV)$ holds.

\subparagraph{An overview of the proof of \cref{main-parikh}}
The remaining part of this section is dedicated to the proof of our second main result, \cref{main-parikh}. The first few steps (\cref{determinize,singlezvass,lem:skeletons,words-to-vectors}) are essentially the same as in \cite{RegSepParikh}, for which we briefly give an overview: The authors reduce regular separability to recognizable separability of semilinear sets in $\N^d$ (for some dimension $d$). Concretely, instead of asking for the regular separability in two given $\Z$-VASS we are counting the cycles within runs of these $\Z$-VASS. Accordingly, the dimension $d$ corresponds to the number of (simple) cycles. Unfortunately, this number is exponential in the size of the input and therefore we cannot just use our first main result (\cref{main-epa}) to prove the $\coNP$ upper complexity bound. Instead we will construct two $\Z$-VASS (of polynomial dimension) accepting sequences of cycles such that their language intersection corresponds to the intersection \eqref{inseparability-intersection2} from \cref{characterization-recognizable-separability2} (which is non-empty if, and only if, the $\Z$-VASS from the input are regularly inseparable). Intersection for $\Z$-VASS is known to be in $\NP$ implying also the $\NP$ upper complexity bound for the regular inseparability problem resp.\ the $\coNP$ upper bound for the separability problem of $\Z$-VASS.

\subparagraph{Reduction to a single integer VASS}
As announced, we will first follow the reduction from \cite{RegSepParikh}. In the first step, the regular separability problem of nondeterministic $\Z$-VASS can be reduced to the same problem in \emph{deterministic} $\Z$-VASS. This reduction is possible in polynomial time which is a bit surprising at first glance since determinization typically requires at least an exponential blowup. However, in this reduction we determinize the $\Z$-VASS ``up to some homomorphic preimage'', i.e., from two given $\Z$-VASS $\cV_1$ and $\cV_2$ one constructs two deterministic $\Z$-VASS $\cW_1$ and $\cW_2$ with
(i) $L(\cW_i)=\inv{h}(L(\cV_i))$ where $h\colon\Gamma^*\to\Sigma^*$ is a homomorphism and (ii) $\regsep{L(\cV_1)}{L(\cV_2)}$ if, and only if, $\regsep{L(\cW_1)}{L(\cW_2)}$ holds. Since our setting is technically slightly different, we include a proof in\ %
\ifFullVersion%
  \cref{app:parikh-automata}.%
\else%
  the full version.%
\fi%
\begin{restatable}[{\cite[Lemma~7]{RegSepParikh}}]{lemma}{determinize}\label{determinize}
  Regular separability for $\Z$-VASS reduces in polynomial time to the regular separability problem for deterministic $\Z$-VASS.
\end{restatable}

Next, we reduce regular separability for deterministic $\Z$-VASS to regular separability of two languages accepted by the same deterministic $\Z$-VASS, but with different sets of counters. To this end, given $d$-dim.\ $\Z$-VASS $\cV_1$ and $\cV_2$ we construct one $2d$-dim.\ $\Z$-VASS $\cV$ (using product construction) and two index sets $I_1,I_2\subseteq[1,2d]$ such that $L(\cV_i)=L(\cV,I_i)$. We include a detailed proof for our setting in\ %
\ifFullVersion%
  \cref{app:parikh-automata}.%
\else%
  the full version.%
\fi%
\begin{restatable}[{\cite[Proposition~1]{RegSepParikh}}]{lemma}{singleZVASS}\label{singlezvass}
  Regular separability for deterministic $\Z$-VASS reduces in polynomial time to the following:
  \begin{itemize}
    \item[] \emph{Given:} A $d$-dimensional deterministic $\Z$-VASS $\cV$ with two subsets $I_1,I_2\subseteq[1,d]$.
    \item[] \emph{Question:} Are the languages $L(\cV,I_1)$ and $L(\cV,I_2)$ regularly separable?
  \end{itemize}
\end{restatable}

\noindent Therefore, we now fix a $\Z$-VASS $\cV=(Q,\Sigma,T,\iota,f)$.

\subparagraph{Skeletons}
Now, we want to further simplify the regular separability problem. Concretely, we want to consider only runs in $\cV$ that are in some sense similar. We consider some base paths---so called \emph{skeletons}---in $\cV$. Two runs in $\cV$ are similar if they follow the same base path and only differ in the order and repetition of some cycles. We define the function $\skel\colon T^*\to T^*$ such that $\skel(r)=\rho$ for a path $r\in T^*$ in $\cV$ such that $\rho$ is a sub-path of the original path $r$ in which we keep the same set of visited states while removing all cycles that do not increase the set of visited states. Here, $\rho$ is called the \emph{skeleton} of $r$.

Let $t_1\cdots t_\ell\in T^*$ be a path in $\cV$, i.e., we have $t_i=(q_{i-1},a_i,\bx_i,q_i)\in T$ for each $1\leq i\leq\ell$. The map $\skel$ is defined inductively as follows: $\skel(\varepsilon)=\varepsilon$ and $\skel(t_1)=t_1$. For $1\leq i<\ell$ assume that $\skel(t_1\cdots t_i)=s_1\cdots s_j$ is already constructed and that $s_1\cdots s_j$ is a path ending in $q_i$. Now we consider the transition $t_{i+1}$. If there is no transition $s_k$ (with $0\leq k\leq j$) such that this transition ends in the state $q_{i+1}$, we set $\skel(t_1\cdots t_it_{i+1})=s_1\cdots s_jt_{i+1}$. Note that $s_1\cdots s_jt_{i+1}$ is a path ending in the state $q_{i+1}$.

Otherwise, let $0\leq k\leq j$ be maximal such that $s_k$ ends in $q_{i+1}$. Then $s_{k+1}\cdots s_jt_{i+1}$ is a cycle in $\cV$ (note that $s_{k+1}$ starts with $q_{i+1}$ since $s_1\cdots s_j$ is a path). If all states occurring in the cycle $s_{k+1}\cdots s_jt_{i+1}$ also occur in the path $s_1\cdots s_k$, then we set $\skel(t_1\cdots t_it_{i+1})=s_1\cdots s_k$, i.e., we omit the cycle $s_{k+1}\cdots s_jt_{i+1}$ in the skeleton. Note that the skeleton $s_1\cdots s_k$ is a path ending in $q_{i+1}$. Otherwise at least one state in the cycle does not occur in the path $s_1\cdots s_k$. In this case, we simply add $t_{i+1}$ resulting in $\skel(t_1\cdots t_it_{i+1})=s_1\cdots s_jt_{i+1}$ where $s_1\cdots s_jt_{i+1}$ is also a path ending in $q_{i+1}$. Note that any skeleton of $\cV$ has length at most quadratic in the number of transitions $|T|$ as shown in \cite[Lemma 10]{RegSepParikh}.

Let $\rho$ be a skeleton. A \emph{$\rho$-cycle} is a cycle that only visits states occurring in $\rho$; a \emph{$\rho$-run} is a run $r\in T^*$ with skeleton $\skel(r)=\rho$ (i.e., $r$ is obtained from $\rho$ by inserting $\rho$-cycles).
We write $L(\cV,I,\rho)$ for the set of all words in $L(\cV,I)$ accepted via $\rho$-runs.
\begin{restatable}[{\cite[Lemma~11]{RegSepParikh}}]{lemma}{skeletons}\label{lem:skeletons}
  We have $\regsep{L(\cV,I_1)}{L(\cV,I_2)}$ if, and only if, $\regsep{L(\cV,I_1,\rho)}{L(\cV,I_2,\rho)}$ holds for every skeleton $\rho$.
\end{restatable}

Although this was essentially shown in \cite[Lemma 11]{RegSepParikh}, our setting is strictly speaking slightly different (e.g.\ we have all short rather than only simple cycles), so we include a detailed proof in\ %
\ifFullVersion%
  \cref{app:parikh-automata}.\ %
\else%
  the full version.\ %
\fi%
Thus, it suffices to show that for a given skeleton $\rho$, one can decide regular inseparability of $L(\cV,I_1,\rho)$ and $L(\cV,I_2,\rho)$ in $\NP$. So, from now on, we fix a skeleton $\rho$ and simply write $L(I_i)$ for $L(\cV,I_i,\rho)$. Since we only consider runs that visit states that occur in $\rho$, we may also assume that $\cV$ consists only of the states occurring on $\rho$.
In particular, we only say \emph{cycle} instead of ``$\rho$-cycle''.

\subparagraph{Counting cycles} We now phrase a characterization of regular
separability from \cite{RegSepParikh} in our setting. It says that regular separability of the languages $L(I_1)$ and $L(I_2)$ is equivalent to recognizable separability of vectors that count cycles. Here, we only count \emph{short} cycles of length at most $|Q|$. This is possible since each cycle can be decomposed into short cycles. In the following, we fix the set $S\subseteq T^{\leq|Q|}$ of all \emph{short} cycles in $\cV$.\footnote{Although \cref{determinize,singlezvass,lem:skeletons,words-to-vectors} are essentially the same as in \cite{RegSepParikh}, we are working with \emph{short cycles}, whereas \cite{RegSepParikh} uses \emph{simple cycles}. This will be crucial later, because short cycles can be guessed on-the-fly in a finite automaton without storing the whole cycle.}

For $I\subseteq[1,d]$, we define: if $t=(p,a,\bx,q)\in T$ is a transition then the \emph{effect} $\Delta_I(t)$ of $t$ to the components in $I$ is $\Delta_I(t)=\restrict{\bx}{I}$, i.e.\ the projection of the counter update $\bx$ to $I$. If $r=t_1t_2\cdots t_\ell\in T^*$ is a path, then the \emph{effect} $\Delta_I(r)$ of $r$ to the components in $I$ is the sum of the effects of all transitions on this path, i.e. $\Delta_I(r)=\sum_{i=1}^\ell\Delta_I(t_i)$. Now, let $\bu\in\N^S$ be a multiset of short cycles. Then the \emph{effect} of $\bu$ to the components in $I$ is $\Delta_I(\bu)=\sum_{c\in S}\bu[c]\cdot\Delta_I(c)$. If $\bv\in\N^T$ is a multiset of transitions, then the \emph{effect} of $\bv$ to the components in $I$ is $\Delta_I(\bv)=\sum_{t\in T}\bv[t]\cdot\Delta_I(t)$. In case of $I=[1,d]$ we will also write $\Delta$ instead of $\Delta_I$. Finally, we define
\[ M(I) = \left\{\bu\in\N^S \,\middle|\, \Delta_I(\rho)+\Delta_I(\bu)=\bzero \right\}\,. \]
Hence, $M(I)$ is the set of multisets of short cycles such that inserting them into $\rho$ would lead to an accepting run with acceptance condition $I\subseteq[1,d]$. Since $M(I)$ is the solution set of linear Diophantine equations, it is hyperlinear (see\ %
\ifFullVersion%
  \cref{app:parikh-automata}\ %
\else%
  the full version\ %
\fi%
for a proof).
\begin{restatable}{observation}{hyperlinear}
  Let $I\subseteq[1,d]$. Then $M(I)$ is hyperlinear, i.e., $M(I)=B+V^*$ for two finite sets $B,V\subseteq\N^S$.
\end{restatable}

The following equivalence between regular separability of the languages $L(I_i)$ and recognizable separability of the (hyperlinear) sets $M(I_i)$ was shown in \cite[Lemma~12]{RegSepParikh}. It is straightforward to adapt it to our situation (see\ %
\ifFullVersion%
  \cref{app:parikh-automata}%
\else%
  the full version%
\fi%
).
\begin{restatable}{lemma}{wordsToVectors}\label{words-to-vectors}
  We have $\regsep{L(I_1)}{L(I_2)}$ if, and only if, $\regsep{M(I_1)}{M(I_2)}$.
\end{restatable}

\subparagraph{Reducing inseparability to intersection}
At this point, our proof deviates from the approach of \cite{RegSepParikh}.
According to \cref{words-to-vectors}, it remains to decide whether $\regsep{M(I_1)}{M(I_2)}$, where $M(I_1)$ and $M(I_2)$ are sets of vectors of dimension $|S|$, which is exponential. In \cref{characterization-recognizable-separability2}, we saw that recognizable separability of vector sets $A+U^*$ and $B+V^*$ reduces to intersection emptiness of $A+U^*+V_J^*$ and $B+V^*+U_J^*$, where $J$ is a subset of the twin-unbounded components. However, the exponential dimension of $M(I_1),M(I_2)$ means a direct translation into existential Presburger arithmetic would incur an exponential blowup. 

Instead, our key observation is that one can reduce inseparability to \emph{intersection emptiness of $\Z$-VASS}: The idea is to encode the intersecting vectors $\bu\in (A+U^*+V_J^*)\cap (B+V^*+U_J^*)$, where $M(I_1)=A+U^*$, $M(I_2)=B+V^*$, as \emph{words containing the participating cycles}. Thus, we guess a subset $J$ of the twin-unbounded components, and then construct in polynomial time two $\Z$-VASS $\cW_1$ and $\cW_2$ such that
\begin{align}
	L(\cW_1)&=\{ \#c_1\#c_2\cdots\# c_m \mid m\in\N,~c_1,\ldots,c_m\in S,~\Phi(c_1,\ldots,c_m)\in A+U^*+V_J^* \}, \label{language-w1} \\
	L(\cW_2)&=\{ \#c_1\#c_2\cdots\# c_m \mid m\in\N,~c_1,\ldots,c_m\in S,~\Phi(c_1,\ldots,c_m)\in B+V^*+U_J^* \}, \label{language-w2}
\end{align}
where for cycles $c_1,\ldots,c_m\in S$, the so-called \emph{Parikh vector}
$\Phi(c_1,\ldots,c_m)\in\N^S$ counts how many times each short cycle occurs in
$c_1,\ldots,c_m$: If $c\in S$, then $\Phi(c_1,\ldots,c_m)[c]$ is the number of
indices $i\in[1,m]$ with $c_i=c$. Note that then clearly, $(A+U^*+V_J^*)\cap
(B+V^*+U_J^*)\ne\emptyset$ if and only if $L(\cW_1)\cap L(\cW_2)\ne\emptyset$.

The main challenge in constructing $\cW_1$ and $\cW_2$ is to guess a subset $J$
of twin-unbounded components, and for the $\Z$-VASS to verify that a given
cycle belongs to $J$, without being able to store an entire cycle in its state.
To solve this, we we will characterize the twin-unbounded cycles in terms of
its set of occurring transitions.

\subparagraph{Characterizing twin-unbounded cycles}
We define for any $\hat{T}\subseteq T$ the set
\[ S[\hat{T}]=\left\{c\in \hat{T}^{\leq|Q|} \,\middle|\, \text{$c$ is a cycle}\right\}. \]
Thus, $S[\hat{T}]\subseteq S$ is the set of all short cycles that consist solely of transitions from $\hat{T}$.

Our characterization uses an adaptation of the notion of ``cancelable productions'' in $\Z$-grammars used in \cite{FOSSACS23}.
We define the homomorphism $\partial\colon \N^T\to\Z^Q$ as follows: for each transition $t=(p,a,\bx,q)\in T$ we set $\partial(\be_t)=\be_q-\be_p$, where $\be_t\in\N^T$ and $\be_p,\be_q\in\N^Q$ are unit vectors.
Thus, $\partial(\bu)[q]$ is the number of incoming transitions to $q$, minus the number of outgoing edges from $q$, weighted by the coefficients in $\bu$. 
A \emph{flow} is a vector $\bmf\in\N^T$ with $\partial(\bmf)=\bzero$.  The
following is a standard fact in graph theory. For a proof that even applies to
context-free grammars (rather than automata), see \cite[Theorem
3.1]{DBLP:journals/fuin/Esparza97}.
\begin{restatable}{lemma}{flows}\label{lem:flows}
	A vector $\bmf\in\N^T$ is a flow if and only if it is a sum of (the Parikh vectors of) cycles.
\end{restatable}

The following notion will be key in characterizing which cycles are twin-unbounded for $M(I_1)$ and $M(I_2)$.
A transition $t\in T$ is \emph{bi-cancelable} if there exist flows $\bmf_1,\bmf_2\in\N^T$  such that\label{def-bicancel}
		(i)~$\Delta_{I_1}(\bmf_1)=\bzero$ and $\Delta_{I_2}(\bmf_2)=\bzero$,
		(ii)~$t$ occurs in both $\bmf_1$ and in $\bmf_2$, and
		(iii)~$\supp(\bmf_1)=\supp(\bmf_2)$.
In other words, $t$ is bi-cancelable if it is part of two flows $\bmf_1$ and $\bmf_2$ with the same support and with effect zero (wrt.\ the components $I_1$ resp.\ $I_2$).

\begin{restatable}{lemma}{strongUnboundBiCancel}\label{lem:strongUnboundBiCancel}
	A cycle $c\in S$ is twin-unbounded for $M(I_1)$ and $M(I_2)$ if, and only if, every
	transition in $c$ is bi-cancelable.
\end{restatable}
\begin{proof}
  For the ``only if'' direction, suppose that $c$ is twin-unbounded for $M(I_1)$ and $M(I_2)$. Then by definition there exist sums of period vectors $\bu_1,\bu_2\in\N^S$ of $M(I_1)$ resp.\ $M(I_2)$ with $c\in\supp(\bu_1)=\supp(\bu_2)$. Define $\bmf_i=\tau(\bu_i)\in\N^T$, where $\tau\colon\N^S\to\N^T$ maps cycles to the number of occurrences of each transition in these cycles. Then clearly $\bmf_i$ are flows with $\Delta_{I_i}(\bmf_i)=\Delta_{I_i}(\bu_i)=\bzero$, $c$ occurs in both $\bmf_1$ and in $\bmf_2$, and $\supp(\bmf_1)=\supp(\bmf_2)$. Hence, all transitions in $c$ are bi-cancelable.
  
	For the ``if'' direction, suppose a cycle $c\in S$ only contains
	bi-cancelable transitions and write $c=t_1\cdots t_n$ for $t_1,\ldots,t_n\in T$. For each $t_i$,
	there are flows $\bmf_{i,1}$ and $\bmf_{i,2}$ witnessing that $t_i$ is
	bi-cancelable. Notice that $\bmf_1:=\bmf_{1,1}+\cdots+\bmf_{n,1}$ and
	$\bmf_2=\bmf_{1,2}+\cdots+\bmf_{n,2}$ are flows as well and they have
	$\supp(\bmf_1)=\supp(\bmf_2)$.
	As flows, both $\bmf_1$ and $\bmf_2$ can be written as a sum of
	cycles: There are $\bu_1,\bu_2\in\N^S$ with
	$\tau(\bu_1)=\bmf_1$ and $\tau(\bu_2)=\bmf_2$. 
	Observe that $\Delta_{I_1}(\bu_1)=\Delta_{I_2}(\bu_2)=\bzero$, meaning $\bu_1$ and $\bu_2$ are sums of period vectors of $M(I_1)$ and $M(I_2)$, respectively. If we knew that $c$ occurs in both $\bu_1$ and in $\bu_2$, and $\bu_1,\bu_2$ had the same support, we could conclude twin-unboundedness of $c$. Since $\bu_1,\bu_2$ may not have these properties, we will now modify them.
	Consider the set $S'=S[\supp(\bmf_1)]=S[\supp(\bmf_2)]$; hence $S'$ is the set of short cycles $u\in T^*$ such that $\supp(u)\subseteq \supp(\bmf_1)=\supp(\bmf_2)$. By the choice of $\bmf_1$ and $\bmf_2$, we know $c\in S'$. For each cycle $u\in S'$, the vectors $\bmf_1-\tau(\be_u)$ and $\bmf_2-\tau(\be_u)$ are again flows, because $\tau(\be_u)$ is a flow. Now observe
	\[ \sum_{u\in S'} \tau(\be_u)+(\bmf_i-\tau(\be_u))=|S'|\cdot\bmf_i \]
	for $i=1,2$ (cf.~\cref{fig:flow}). Hence, the flow $|S'|\cdot\bmf_i$ can be written as a sum of cycles in which each cycle from $S'$ occurs. Moreover, in this sum, every occurring cycle belongs to $S'$. This means, $\bu'_1,\bu'_2$ have the same support $S'$, which includes $c$. Moreover, since $\tau(\bu'_i)=|S'|\cdot\bmf_i$, we know that $\Delta_{I_i}(\bu'_i)=\bzero$, meaning $\bu'_i$ is a sum of period vectors of $M(I_i)$, for $i=1,2$.
  This means, $c$ is indeed twin-unbounded for $M(I_1)$ and $M(I_2)$.\qedhere
  
  \begin{figure}
    \begin{center}
      \begin{tikzpicture}[>=Latex,every node/.style={draw,circle,fill=black,inner sep=1pt}]
        \foreach \a in {0,60,...,300} {
          \node (\a) at (\a:1cm) {};
        }
        \node (a1) at (2,0.5) {};
        \node (a2) at (2,-0.5) {};
        \node (a3) at (2.5,0) {};
        
        \draw[->,thick,blue]
          (60) edge (0)
          (120) edge (60)
          (180) edge (120)
          (240) edge (180)
          (300) edge (240)
          (0) edge (300);
        
        \draw[->,red,every edge/.append style={bend right,looseness=2,out=-90,in=-90},dashed] 
          (0)   edge (60)
          (60)  edge (120)
          (120) edge (180)
          (180) edge (240)
          (240) edge (300)
          (300) edge (0);
        
        \draw[->,red,every edge/.append style={loop,looseness=20,min distance=40},dashed]
          (a1) edge[out=135,in=45] (a1)
          (a2) edge[out=-45,in=-135] (a2)
          (a3) edge[out=45,in=-45] (a3);
      \end{tikzpicture}
    \end{center}
    \caption{The flow $\tau(\be_u)+(\bmf_i-\tau(\be_u))$ where the cycle $u$ is depicted in bold blue and the cycles of the flow $\bmf_i-\tau(\be_u)$ are depicted in red. Note that the new flower shaped cycle is not necessarily short, but can be easily split into short cycles.\label{fig:flow}}
  \end{figure}
\end{proof}

To construct our $\Z$-VASS $\cW_1$ and $\cW_2$, we first guess a set of
transitions and then verify that all of them are bi-cancelable. For the
verification, we translate the definition of bi-cancelability into an existential
Presburger formula $\varphi_t$ which is satisfiable if, and only if, $t$ is
bi-cancelable (see\ %
\ifFullVersion%
  \cref{app:parikh-automata}%
\else%
  the full version%
\fi%
).
\begin{restatable}{lemma}{biCancelNP}\label{lem:bicancelNP}
	Given a transition $t\in T$, we can decide in $\NP$ whether it is bi-cancelable.
\end{restatable}

\subparagraph{Constructing the $\Z$-VASS}\label{par:emptiness}
Let us now describe in more detail how the $\Z$-VASS $\cW_1$ and $\cW_2$ are constructed.
Instead of literally guessing the set $J$ of twin-unbounded cycles (which could require exponentially many bits), we guess a set $\hat{T}\subseteq T$ of transitions in $\cV$ and then verify in $\NP$ that they are all bi-cancelable using \cref{lem:bicancelNP}. Then, we build $\Z$-VASS that satisfy \cref{language-w1,language-w2} for the specific choice $J=S[\hat{T}]$. This means, we will have
\begin{align}
	L(\cW_1)&=\{ \#c_1\#c_2\cdots\# c_m \mid m\in\N, c_1,\ldots,c_m\in S, \Phi(c_1,\ldots,c_m)\in A+U^*+V_{S[\hat{T}]}^* \} \label{language-w1-specific} \\
	L(\cW_2)&=\{ \#c_1\#c_2\cdots\# c_m \mid m\in\N, c_1,\ldots,c_m\in S, \Phi(c_1,\ldots,c_m)\in B+V^*+U_{S[\hat{T}]}^* \} \label{language-w2-specific}
\end{align}
and from now on, we will also write $J=S[\hat{T}]$.
Note that the result of our algorithm is correct, even when the guess for $\hat{T}$ is not the \emph{entire} set of bi-cancelable transitions: when $L(\cW_1)$ intersects $L(\cW_2)$ for some choice of $\hat{T}$, it will do so for any larger choice of $\hat{T}$.

\subparagraph{Ensuring membership in $\boldsymbol{A+U^*}$}
The idea for constructing $\cW_1$ (and analogously $\cW_2$) is simple. For each cycle in the input, it guesses whether it belongs to $A+U^*$ or to $V_{S[\hat{T}]}^*$. Let $\bu_0\in\N^S$ and $\bu_1\in\N^S$ be the collection of cycles guessed to be in $A+U^*$ and in $V_{S[\hat{T}]}^*$, respectively. To make sure that $\bu_0\in A+U^*$, we note that $\bu_0\in A+U^*$ is equivalent to $\Delta_{I_1}(\bu_0)+\Delta_{I_1}(\rho)=\bzero$, where $\rho$ is the skeleton guessed earlier in the algorithm. Thus, we can use $|I_1|$ counters to sum up the effect of the cycles $\bu_0$ and add $\Delta_{I_1}(\rho)$ once in the end. Hence, these counters being zero in the end is equivalent to $\bu_0\in A+U^*$.

\subparagraph{Ensuring membership in $\boldsymbol{V_{S[\hat{T}]}^*}$}
To make sure that $\bu_1\in V_{S[\hat{T}]}^*$, we note that this is equivalent to $\Delta_{I_2}(\bu_1)=\bzero$ and $\supp(\bu_1)\subseteq S[\hat{T}]$. Thus, our $\Z$-VASS has a separate set of $|I_2|$ counters that carry the total effect of all the cycles in $\bu_1$. Moreover, it is easy to check that all cycles in $\bu_1$ only use transitions in $\hat{T}$.

Note that membership in $B+V^*$ and in $U_{S[\hat{T}]}^*$ are checked similarly.
With this polynomial-time construction of $\cW_1$ and $\cW_2$, we are ready to prove \cref{main-parikh}:
\begin{proof}[Proof of \cref{main-parikh}]
	We give an $\NP$ algorithm for regular inseparability of two $\Z$-VASS (which
	can be obtained from Parikh automata in logarithmic space
	\cite[Corollary~1]{FOSSACS23}).
  
	Let $\cV_1$ and $\cV_2$ be two $d$-dimensional $\Z$-VASS. From $\cV_1$ and
	$\cV_2$ we can compute a single $2d$-dimensional deterministic
	$\Z$-VASS $\cV$ and two sets $I_1,I_2\subseteq[1,2d]$ in polynomial
	time such that $\regsep{L(\cV_1)}{L(\cV_2)}$ holds if, and only if,
	$\regsep{L(\cV,I_1)}{L(\cV,I_2)}$ (\cref{determinize,singlezvass}).
	According to \cref{lem:skeletons} we have
	$\regsep{L(\cV,I_1)}{L(\cV,I_2)}$ if, and only if,
	$\regsep{L(\cV,I_1,\rho)}{L(\cV,I_2,\rho)}$ for each skeleton $\rho$ in
	$\cV$ holds. So, we guess a skeleton $\rho$ and check regular
	inseparability of $L(\cV,I_1,\rho)$ and $L(\cV,I_2,\rho)$ certifying
	regular inseparability of $L(\cV,I_1)$ and $L(\cV,I_2)$.
  
	Additionally, we will guess a set $\hat{T}\subseteq T$ of transitions
	and verify in $\NP$ that all of them are bi-cancelable
	(\cref{lem:bicancelNP}). Then we can construct in polynomial time two
	$\Z$-VASS $\cW_1$ and $\cW_2$ such that \eqref{language-w1-specific}
	and \eqref{language-w2-specific} hold. If $L(\cW_1)\cap L(\cW_2)\ne\emptyset$, the algorithm reports ``inseparable''. For this, it uses
	a simple product construction to obtain a $\Z$-VASS $\cW$ for the
	intersection $L(\cW_1)\cap L(\cW_2)$, and decide in $\NP$ whether an
	accepting configuration can be reached in $\cW$.

	This is sound: We have $L(\cW_1)\cap
	L(\cW_2)\ne\emptyset$ if and only if $(A+U^*+V_J^*)\cap
	(B+V^*+U_J^*)\ne\emptyset$ for $J=S[\hat{T}]$; and by
	\cref{words-to-vectors}, we know that the latter rules out
	$\regsep{M(I_1)}{M(I_2)}$. For completeness, note that if $\regsep{M(I_1)}{M(I_2)}$ does not hold,
	then there exists a choice for $\hat{T}$ such that $L(\cW_1)\cap
	L(\cW_2)\ne\emptyset$: Take the set of all bi-cancelable transitions.
\end{proof}

\label{beforebibliography}
\newoutputstream{pages}
\openoutputfile{main.pages.ctr}{pages}
\addtostream{pages}{\getpagerefnumber{beforebibliography}}
\closeoutputstream{pages}

\bibliography{main.bib}

\ifFullVersion
\newpage

\begin{appendix}
\section{Omitted proofs of Section~\ref{sec:hyperlinear}}\label{app:hyperlinear}
\subgroupSeparability*
\begin{proof}
  Consider the quotient group $\Z^d/A$. It is finitely generated and abelian and thus isomorphic to a group $\bigoplus_{j=1}^n \Z/r_j\Z$ for some numbers $r_1,\ldots,r_n\in\N$. The projection map $\pi\colon \Z^d\to\Z^d/A$ can thus be composed with the isomorphism above to yield a morphism $\psi\colon \Z^d\to\bigoplus_{j=1}^n\Z/r_j\Z$ with $\ker\psi=A$. Since $\bu\notin A$ and thus $\psi(\bu)\ne 0$, say the $j$-th component of $\psi(\bu)$ is not zero. We distinguish two cases:
  \begin{enumerate}[(1)]
    \item If $r_j>0$, then we can choose $\varphi\colon\Z^d\to\Z/r_j\Z$ to be $\psi$ followed by the projection to the $j$-th component.
    \item If $r_j=0$, then $\Z/r_j\Z=\Z$ and thus the $j$-th component of $\psi(\bu)$ is an integer $k\in\Z$. We pick some $s>|k|$ and let $\varphi\colon\Z^d\to\Z/s\Z$ yield the $j$-th component of $\psi$, modulo $s$.
  \end{enumerate}
  These choices clearly satisfy $\varphi(A)=0$ and $\varphi(\bu)\ne 0$.
\end{proof}

\section{Omitted proofs of Section~\ref{section-epa}}\label{app:epa}
\lemRepSepUnion*
\begin{proof}
  Assume $\regsep{K}{L}$. Then there is a recognizable sets $S\subseteq M$
  separating $K$ and $L$. Let $1\leq i\leq n$ be arbitrary. Since
  $K_i\subseteq K$ holds, the set $S$ is also a separator of $K_i$
  and $L$, i.e., $\regsep{K_i}{L}$ for all $1\leq i\leq n$.
  
  Conversely, assume $\regsep{K_i}{L}$ for all $1\leq i\leq n$. Then
  there are recognizable sets $S_i\subseteq M$ separating $K_i$ and $L$.
  Set $S:=\bigcup_{1\leq i\leq n}S_i$. Then $S$ is recognizable (according
  to the closure properties of the class of recognizable sets). We also have
  \[K=\bigcup_{1\leq i\leq n}K_i\subseteq\bigcup_{1\leq i\leq n}S_i=S\]
  and
  \[
    L\cap S = L\cap\left(\bigcup_{1\leq i\leq n}S_i\right)
            = \bigcup_{1\leq i\leq n}(L\cap S_i)
            = \bigcup_{1\leq i\leq n}\emptyset
            = \emptyset\,.
  \]
  In other words, we $S$ is a recognizable separator of $K$ and $L$,
  i.e., $\regsep{K}{L}$.
\end{proof}

\inseparabilityViaInequalities*
\begin{proof}
	We apply \cref{characterization-recognizable-separability2}. To this
	end, we use the standard hyperlinear representation for solution sets
	of systems of linear Diophantine equalities. Let $A_0\subseteq\N^r$
	be the set of all (component-wise) minimal solutions to $A\bx=\bb$, and let
	$U\subseteq\N^r$ be the set of all minimal solutions to
	$A\bx=\bzero$. Then it is well known that
	$K=\pi(A_0+U^*)=\pi(A_0)+\pi(U)^*$. In the same way, we obtain a hyperlinear
	representation $L=\pi(B_0+V^*)=\pi(B_0)+\pi(V)^*$.
	Then, the
  \lcnamecref{inseparability-via-inequalities} follows from
  \cref{characterization-recognizable-separability2}.
  
  Indeed, observe that then $\pi(U)^*$ is exactly the set of $\pi(\bp)\in\N^d$
  with $A\bp=\bzero$. Likewise, $\pi(V)^*$ is exactly the set of
  $\pi(\bq)\in\N^d$ with $C\bq=\bzero$. Therefore, if $J\subseteq[1,d]$
  is the set of twin-unbounded components of $K,L$, and $U_J,V_J$ are
  defined as in \cref{characterization-recognizable-separability2}, then
  $\pi(U_J)^*$ consists of exactly those $\pi(\bu)$ for which (i)~there
  are $\bp,\bq\in\N^r$ with $A\bp=\bzero$ and $C\bq=\bzero$ with
  $\supp(\pi(\bu))\subseteq\supp(\pi(\bp))=\supp(\pi(\bq))\subseteq J$, and
  (ii)~$A\bu=\bzero$. The set $\pi(V_J)^*$ has an analogous description.
  
  Thus, if $\bp,\bq,\bu,\bv,\bx,\by\in\N^r$ exist as in the
  \lcnamecref{inseparability-via-inequalities}, then clearly
  $\pi(\bx+\bv)=\pi(\by+\bu)$ lies in the intersection
  $(\pi(A_0)+\pi(U)^*+\pi(V_J)^*)\cap (\pi(B_0)+\pi(V)^*+\pi(U_J)^*)$. 
  
  Conversely, an element in the intersection
  $(\pi(A_0)+\pi(U)^*+\pi(V_J)^*)\cap (\pi(B_0)+\pi(V)^*+\pi(U_J)^*)$ can
  be written as $\pi(\bx+\bv)=\pi(\by+\bu)$, such that $A\bx=\bb$,
  $C\by=\bd$, and there are $\bp_1,\bq_1\in\N^r$ witnessing
  $\bu\in U_J^*$ and also $\bp_2,\bq_2\in\N^r$ witnessing
  $\bv\in V_J^*$. This means,
  $\supp(\pi(\bu))\subseteq\supp(\pi(\bp_1))=\supp(\pi(\bq_1))$ and
  $A\bp_1=\bzero$ and $C\bq_1=\bzero$, but also
  $\supp(\bv)\subseteq\supp(\pi(\bp_2))=\supp(\pi(\bq_2))$ and
  $A\bp_2=\bzero$ and $C\bq_2=\bzero$. But then we can use
  $\bp:=\bp_1+\bp_2$ and $\bq:=\bq_1+\bq_2$ to satisfy the requirements
  of the \lcnamecref{inseparability-via-inequalities}.
\end{proof}

\section{Omitted proofs of Section~\ref{section-parikh-automata}}\label{app:parikh-automata}

\begin{lemma}\label{lem:alphMorph}
  Let $K,L\subseteq\Sigma^*$ be two languages and $h\colon\Gamma^*\to\Sigma^*$ be an alphabetic morphism\footnote{A morphism $h\colon\Gamma^*\to\Sigma^*$ is \emph{alphabetic} if $|h(a)|\leq1$ holds for each letter $a\in\Gamma$.}. If $K'\subseteq\inv{h}(K)$ with $h(K')=K$, then we have
  \[\regsep{K}{L} \iff \regsep{K'}{\inv{h}(L)}\,.\]
\end{lemma}
\begin{proof}
  First, assume $\regsep{K}{L}$. Then there is a regular separator $R\subseteq\Sigma^*$ of $K$ and $L$, i.e., we have $K\subseteq R$ and $L\cap R=\emptyset$. Set $R':=\inv{h}(R)\subseteq\Gamma^*$. $R'$ is regular since the class of regular languages is closed under inverse morphisms. We also have $K'\subseteq\inv{h}(K)\subseteq\inv{h}(R)=R'$. Additionally, we have $\inv{h}(L)\cap\inv{h}(R)=\emptyset$ since the existence of an element $w\in \inv{h}(L)\cap\inv{h}(R)$ would imply $h(w)\in L\cap R$. This means, $R'$ is a regular separator of $K'$ and $\inv{h}(L)$, i.e., $\regsep{K'}{\inv{h}(L)}$.
  
  Conversely, assume $\regsep{K'}{\inv{h}(L)}$. Then there exists a regular separator $R'\subseteq\Gamma^*$ of $K'$ and $\inv{h}(L)$, i.e., we have $K'\subseteq R'$ and $\inv{h}(L)\cap R'=\emptyset$. Set $R:=h(R')$ which is a regular language since the class of regular languages is also closed under morphisms. Then we have $K=h(K')\subseteq h(R')=R$. Also $L\cap R=\emptyset$ holds: towards a contradiction suppose there is $w\in L\cap R$. From $w\in R=h(R')$ follows the existence of a word $w'\in R'$ with $h(w')=w$. We also infer $w'\in\inv{h}(L)$ from $w\in L$. Hence, we have $w'\in\inv{h}(L)\cap R'=\emptyset$---a contradiction. All in all, we proved that $R$ is a regular separator of $K$ and $L$, i.e., $\regsep{K}{L}$.
\end{proof}

\determinize*
\begin{proof}
  The proof of this lemma is similar to \cite[Lemma~7]{RegSepParikh}: let $\cV_i=(Q_i,\Sigma,T_i,\iota_i,f_i)$ with $i=1,2$ be two $\Z$-VASS. From $\cV_1$ and $\cV_2$ we will construct two $\Z$-VASS $\cV_i'=(Q_i,\Gamma,T_i',\iota_i,f_i)$ such that $\cV_1'$ is deterministic and we have
  \[\regsep{L(\cV_1)}{L(\cV_2)} \iff \regsep{L(\cV_1')}{L(\cV_2')}\,.\]
  We will obtain the determinism of $\cV_1'$ by making each label of a transition in $\cV_1$ unique. So, set $\Gamma=T_1$. $T_1'$ is obtained from $T_1$ by replacing each transition $t=(p,a,\bx,q)\in T_1$ by the new transition $(p,t,\bx,q)$. Using this translation we also obtain a morphism $h\colon\Gamma^*\to\Sigma^*$ with $h((p,a,\bx,q))=a$ for each transition $(p,a,\bx,q)\in\Gamma=T_1$. Then we obtain $\cV_2'$ from $\cV_2$ with $L(\cV_2')=\inv{h}(L(\cV_2))$ by replacing each label $a\in\Sigma_\varepsilon$ of a transition in $T_2'$ with all labels $t\in T_1$ with $h(t)=a$. Additionally, we add loops labeled with $t\in T_1$ such that $h(t)=\varepsilon$ to any state of $L(\cV_2)$. Formally, this is the following set of transitions:
  \begin{align*}
    T_2'=&\phantom{{}\cup{}}\{(p,t,\bx,q)\mid t\in T_1,(p,h(t),\bx,q)\in T_2\}\\
         &\cup\{(p,t,\bzero,q)\mid p,q\in Q,t\in T_1,h(t)=\varepsilon\}\,.
  \end{align*}
  Note that this is a well known construction for the application of the inverse of an alphabetic morphism and, hence, we have $L(\cV_2')=\inv{h}(L(\cV_2))$.
  
  Since each letter from $\Gamma$ occurs in exactly one transition of $\cV_1'$, this $\Z$-VASS is deterministic. Additionally, $\cV_1'$ and $\cV_2'$ can be constructed from $\cV_1$ and $\cV_2$ in polynomial time. It is also clear that the morphism $h$ is alphabetical. We can also prove the following properties:
  \begin{enumerate}
    \item $L(\cV_1')\subseteq\inv{h}(L(\cV_1))$: Let $w\in L(\cV_1')$. Then there is an accepting run $t_1't_2'\cdots t_\ell'$ in $\cV_1'$ with $t_i=(q_{i-1},t_i,\bx_i,q_i)\in T_1'$ for each $1\leq i\leq\ell$. In particular, we have $w=t_1t_2\cdots t_\ell\in T_1^*$. By definition of $\cV_1'$ we have $t_i=(q_{i-1},a_i,\bx_i,q_i)\in T_1$ for an $a_i\in\Sigma_\varepsilon$. But this means that $w=t_1t_2\cdots t_\ell$ is an accepting run in $\cV_1'$ labeled by $a_1a_2\cdots a_\ell$, i.e., $a_1a_2\cdots a_\ell\in L(\cV_1)$. Moreover, we have $h(w)=h(t_1t_2\cdots t_\ell)=a_1a_2\cdots a_\ell$ implying $w\in\inv{h}(a_1a_2\cdots a_\ell)\subseteq\inv{h}(L(\cV_1))$.
    \item $h(L(\cV_1'))=L(\cV_1)$: A word $w\in\Sigma^*$ is in $h(L(\cV_1'))$ if, and only if, there is a word $w'\in L(\cV_1')\subseteq\Gamma^*$ with $w=h(w')$. This is exactly the case if there is an accepting run $t_1't_2'\cdots t_\ell'$ in $\cV_1'$ that is labeled with $w'$, i.e., we have $t_i'=(q_{i-1},t_i,\bx_i,q_i)\in T_1'$ and $w'=t_1t_2\cdots t_\ell$. By construction this is equivalent to an accepting run $t_1t_2\cdots t_\ell$ in $\cV_1$ that is labeled with $h(w')=w$. But this is exactly the definition of $w\in L(\cV_1)$.
  \end{enumerate}
  Now, we can apply \cref{lem:alphMorph} and obtain
  \[\regsep{L(\cV_1)}{L(\cV_2)} \iff \regsep{L(\cV_1')}{L(\cV_2')}\,.\]
  
  In a final step, we can apply the same polynomial-time procedure to $\cV_2'$ and $\cV_1'$ to determinize $\cV_2'$. The result are two $\Z$-VASS $\cV_1''$ and $\cV_2''$ with
  \[\regsep{L(\cV_1)}{L(\cV_2)} \iff \regsep{L(\cV_1')}{L(\cV_2')} \iff \regsep{L(\cV_1'')}{L(\cV_2'')}\,.\]
  While $\cV_2''$ is deterministic by construction, it is not clear that the same holds for $\cV_1''$. However, due to the fact that $\cV_1'$ and $\cV_2'$ do not have any $\varepsilon$-transitions, the resulting morphism $h'\colon T_2'^*\to T_1^*$ is strictly alphabetical. Hence, $\cV_1''$ is also deterministic.
\end{proof}

\singleZVASS*
\begin{proof}
  Let $\cV_i=(Q_i,\Sigma,T_i,\iota_i,f_i)$ be two deterministic $d$-dimensional $\Z$-VASS. We apply the product construction and obtain a new deterministic $2d$-dimensional $\Z$-VASS $\cV_1\times\cV_2=(Q_1\times Q_2,\Sigma,T,(\iota_1,\iota_2),(f_1,f_2))$ with
  \[
    T=\left\{((p_1,p_2),a,(\bv_1,\bv_2),(q_1,q_2))\,\middle|\,\begin{matrix}
      (p_i,a,\bv_i,q_i)\in T_i\\
      \text{ for all }i=1,2
    \end{matrix}\right\}.
  \]
  
  We show now that $\regsep{L(\cV_1)}{L(\cV_2)}$ holds if, and only if, $\regsep{L(\cV_1\times\cV_2,[1,d])}{L(\cV_1\times\cV_2,[d+1,2d])}$.
  Let $\cA_i=(Q_i,\Sigma,\Delta_i,\iota_i,\{f_i\})$ with $\Delta_i=\{(p,a,q)\mid\exists\bv\in\Z^d\colon(p,a,\bv,q)\in T_i\}$ be the DFA obtained from $\cV_i$ (for $i=1,2$) by removing all counter updates from the transitions. Then we can observe that $L(\cV_1\times\cV_2,[1,d])=L(\cV_1)\cap L(\cA_2)$ and $L(\cV_1\times\cV_2,[d+1,2d])=L(\cV_2)\cap L(\cA_1)$ holds.
  
  Assume that $\regsep{L(\cV_1)}{L(\cV_2)}$ holds. Then there is a regular separator $R\subseteq\Sigma^*$ with $L(\cV_1)\subseteq R$ and $L(\cV_2)\cap R=\emptyset$. Since $L(\cV_1\times\cV_2,[1,d])=L(\cV_1)\cap L(\cA_2)\subseteq L(\cV_1)$ and, similarly, $L(\cV_1\times\cV_2,[d+1,2d])\subseteq L(\cV_2)$ holds, the regular language $R$ is also a separator of $L(\cV_1\times\cV_2,[1,d])$ and $L(\cV_1\times\cV_2,[d+1,2d])$.
  
  Conversely, let $R\subseteq\Sigma^*$ be a regular separator of $L(\cV_1\times\cV_2,[1,d])$ and $L(\cV_1\times\cV_2,[d+1,2d])$. Set $R'=(R\cap L(\cA_1))\cup(\Sigma^*\setminus L(\cA_2))$. Clearly the language $R'$ is regular. We also have
  \begin{align*}
    L(\cV_1) &= (L(\cV_1)\cap L(\cA_2))\cup(L(\cV_1)\cap\Sigma^*\setminus L(\cA_2))\\
    &= (L(\cV_1)\cap L(\cA_2)\cap L(\cA_1))\cup(L(\cV_1)\cap\Sigma^*\setminus L(\cA_2))\\
    &\subseteq (R\cap L(\cA_1))\cup(L(\cV_1)\cap\Sigma^*\setminus L(\cA_2))\\
    &\subseteq (R\cap L(\cA_1))\cup(\Sigma^*\setminus L(\cA_2))\\
    &= R'\,.
  \end{align*}
  Here, the second line holds since $L(\cV_1)\subseteq L(\cA_1)$ and the third one holds since $R$ is a separator.
  
  Additionally, by $L(\cV_2)\subseteq L(\cA_2)$ we have $L(\cV_2)\cap(\Sigma^*\setminus L(\cA_2))=\emptyset$ and
  \[(R\cap L(\cA_1))\cap L(\cV_2)=R\cap L(\cV_1\times\cV_2,[d+1,2d])=\emptyset\]
  implying $L(\cV_2)\cap R'=\emptyset$. Hence, $R'$ is a regular separator of $L(\cV_1)$ and $L(\cV_2)$.
\end{proof}

\skeletons*
\begin{proof}
  First, note that there are only finitely many skeletons: Clemente et al.\ proved in \cite[page~9]{RegSepParikh} that each skeleton has length at most $|Q|^2$. Hence, there are at most $|T|^{|Q|^2}$ many skeletons in $\cV$. It is also clear that $L(\cV,I)=\bigcup_{\text{skeleton $\rho$ of $\cV$}}L(\cV,I,\rho)$ holds.
  
  Let $\rho$ be a skeleton of $\cV$. There is also a regular language $K_\rho\subseteq\Sigma^*$ such that $L(\cV,I,\rho)=L(\cV,I)\cap K_\rho$ holds: we can obtain a finite automaton accepting $K_\rho$ from $\cV$ and $\rho$ by removing all counters and all edges and states that do not belong the skeleton $\rho$.
  
  Finally, we use the following well known fact:
  \begin{claim}
    Let $K_1,\ldots,K_n\subseteq\Sigma^*$ be regular languages partitioning $\Sigma^*$ and $L_1,L_2\subseteq\Sigma^*$ be two languages. Then we have $\regsep{L_1}{L_2}$ if, and only if, $\regsep{L_1\cap K_i}{L_2\cap K_i}$ holds for each $1\leq i\leq n$.
  \end{claim}
  Now, if the languages $K_i$ are the regular languages $K_\rho$ for any skeleton $\rho$ and $L_i=L(\cV,I_i)$ for $i=1,2$ we obtain that $\regsep{L(\cV,I_1)}{L(\cV,I_2)}$ holds if, and only if, $L(\cV,I_1,\rho)=L(\cV,I_1)\cap K_\rho$ is regular separable from $L(\cV,I_2)\cap K_\rho=L(\cV,I_2,\rho)$.
\end{proof}

\wordsToVectors*
\begin{proof}
  Before we prove the equivalence, let us introduce a map $\cycles\colon T^*\to\N^S$ such that for each $\rho$-run $r\in T^*$ we have $\cycles(r)=\bv\in\N^S$ if $r$ contains each $\rho$-cycle $c\in S$ exactly $\bv[c]$ times.
  
  Now, assume that $\regsep{L(I_1)}{L(I_2)}$ holds, i.e., there is a regular separator $R\subseteq\Sigma^*$ with $L(I_1)\subseteq R$ and $R\cap L(I_2)=\emptyset$. We will use \cref{lem:insepSim} to show that $M(I_1)$ and $M(I_2)$ are separable by a recognizable set. To this end, we will give a number $k\in\N\setminus\{0\}$ such that $\bv_1\nsim_k\bv_2$ holds for each $\bv_i\in M(I_i)$ implying the separability of $M(I_1)$ and $M(I_2)$.
  
  For two words $w_1,w_2\in\Sigma^*$ write $w_1\equiv_R w_2$ if $xw_1y\in R\iff xw_2y\in R$ for all $x,y\in\Sigma^*$ (i.e., $\equiv_R$ is the \emph{syntactic} or \emph{Myhill congruence} of $R$). Since $R$ is regular, the index of $\equiv_R$ is finite and, hence, there is a number $k\in\N\setminus\{0\}$ such that
  \begin{equation}
    w^k\equiv_R w^{2k} \quad \text{for each $w\in\Sigma^*$.}\label{eq:myhillNerode}
  \end{equation}
  We show now $\bv_1\nsim_k\bv_2$ for each $\bv_i\in M(I_i)$. Towards a contradiction, assume there are $\bv_i\in M(I_i)$ (for $i=1,2$) with $\bv_1\sim_k\bv_2$. We construct runs $r_i\in T^*$ such that $\skel(r_i)=\rho$ and $\cycles(r_i)=\bv_i$ hold. For a short $\rho$-cycle $c\in S$ choose a prefix $x_c$ of $\rho$ such that $\skel(x_cc)=x_c$ (note that for each cycle $c\in S$ such an $x_c$ exists). Let $c_1,\ldots,c_n$ be an enumeration of $S$ such that $|x_{c_1}|\leq|x_{c_2}|\leq\cdots\leq|x_{c_n}|$ holds. In the following we will write $x_i$ instead of $x_{c_i}$. Let $z_1,\ldots,z_{n+1}\in T^*$ such that $z_1=x_1$, $x_iz_{i+1}=x_{i+1}$ for each $1\leq i<n$, and $x_nz_{n+1}=\rho$, i.e., we have $\rho=z_1z_2\cdots z_{n+1}$. Set
  \[r_i:=z_1c_1^{\bv_i[c_1]}z_2c_2^{\bv_i[c_2]}\cdots z_nc_n^{\bv_i[c_n]}z_{n+1}\,.\]
  Clearly we have $\skel(r_i)=\rho$ and $\cycles(r_i)=\bv_i$ hold for $i=1,2$. We can also show that the labels $w_1,w_2\in\Sigma^*$ of the paths $r_1$ resp.\ $r_2$ satisfy $w_1\equiv_R w_2$ using $\bv_1\sim_k\bv_2$ and repeated usage of the equation \eqref{eq:myhillNerode}. However, $\bv_i\in M(I_i)$ implies $w_i\in L(I_i)$. Since $w_1\in L(I_1)\subseteq R$ we also have $w_2\in R$ (by $w_1\equiv_Rw_2$). Hence, we have $w_2\in R\cap L(I_2)=\emptyset$---a contradiction.
  
  Conversely, assume that $\regsep{M(I_1)}{M(I_2)}$ holds. Hence, there is a recognizable set $X\subseteq\N^S$ such that $M(I_1)\subseteq X$ and $X\cap M(I_2)=\emptyset$. Let $R\subseteq\Sigma^*$ be the set of all labels of $\rho$-runs $r\in T^*$ such that $\skel(r)=\rho$ with $\cycles(r)\in X$. We show that $R$ is a regular separator of $L(I_1)$ and $L(I_2)$. We have $L(I_1)\subseteq R$: let $w\in L(I_1)$. Then $w$ is the label of a $\rho$-run $r\in T^*$ with $\skel(r)=\rho$. But then we know $\cycles(r)\in M(I_1)\subseteq X$ implying $w\in R$.
  
  Now, suppose there is a word $w\in L(I_2)\cap R$. Then $w$ is the label of runs $r_1,r_2\in T^*$ with $\skel(r_i)=\rho$, $\cycles(r_1)\in M(I_2)$ and $\cycles(r_2)\in X$. Since $\cV$ is deterministic, we know that $r_1=r_2$ implying $\cycles(r_1)=\cycles(r_2)\in M(I_2)\cap X=\emptyset$---a contradiction. Hence, we have $L(I_2)\cap R=\emptyset$.
  
  Finally, we have to show that $R$ is regular. To this end, we construct a nondeterministic finite automaton that simulates $\rho$-runs by storing the image of the map $\skel$ in its state. While the set of all skeletons is finite, the set of vectors appearing in the image of $\skel$ is not necessarily bounded. However, since $X$ is recognizable and, hence, semilinear we can evaluate the condition $\cycles(r)\in X$ for a path $r\in T^*$ using only a finite memory. Concretely we guess a linear set $\bu+P^*\subseteq X$ where $\bu\in\N^S$ and $P\subseteq\N^S$ finite (recall that $X$ is a finite union of such linear sets). Additionally, let $P=\{\bp_1,\ldots,\bp_n\}$. The NFA stores in its memory vectors $\bu',\bp_1',\ldots,\bp_n'$ with $\bu'\leq\bu$ and $\bp_i'\leq\bp_i$ for all $1\leq i\leq n$. Whenever the simulation of $\skel$ detects a $\rho$-cycle, we increase one of the vectors $\bu',\bp_1',\ldots,\bp_n'$. If we reach the vector $\bp_i$ due to this increasing, we reset this vector to $\bzero$. The NFA accepts if its memory contains the skeleton $\rho$ and the (bounded) counter values $\bu,\bzero,\ldots,\bzero$. Clearly, this NFA accepts the language $R$. Hence, $R$ is a regular separator of $L(I_1)$ and $L(I_2)$.
\end{proof}

\hyperlinear*
\begin{proof}
  The equation $\Delta_{I_i}(\rho)+\Delta_{I_i}(\bu)=\bzero$ is a system of linear equations (over $\N^S$) and $M(I)$ is the set of solutions of this equation system. Since the equations are expressible in Presburger arithmetic, we obtain that $M(I)$ is semilinear \cite{GinS66a}. Hence, we have $M(I)=\bigcup_{1\leq i\leq k}\bu_i+V_i^*$ (where $\bu_i\in\N^S$ and $V_i\subseteq\N^S$ are finite). We can see that the vectors in $V_i$ are solutions of the homogeneous linear equation system $\Delta_{I_i}(\bv)=\bzero$ and the vectors $\bu_j$ satisfy the inhomogeneous system $\Delta_{I_i}(\bu_j)=-\Delta_{I_i}(\rho)$. Therefore, we have $\bu_i+\bv\in M(I)$ for each $1\leq i\leq k$ and $\bv\in\bigcup_{1\leq j\leq k}V_j^*$. According to this we can write the solution set $M(I)$ also as $B+V^*$ where $B=\{\bu_1,\ldots,\bu_k\}$ and $V=\bigcup_{1\leq i\leq k}V_i$. In other words, the set $M(I)$ is even hyperlinear.
\end{proof}

\biCancelNP*
\begin{proof}
	We construct an existential Presburger formula $\varphi_t$ which is satisfiable if, and only if, $t$ is bi-cancelable. Recall that $t$ is bi-cancelable if, and only if, there exist two flows $\bmf_1,\bmf_2\in\N^T$ such that the properties (i)--(iii) on page~\pageref{def-bicancel} hold. We express in the following these three properties as quantifier-free Presburger formulas using the variables $x_{t'}$ and $y_{t'}$ for each transition.
  \begin{enumerate}[(i)]
    \item $\psi_{1}=\bigwedge_{i\in[1,d]}\sum_{t'=(p,a,\bv,q)\in T}\bv[i]\cdot x_{t'}=0 \land \sum_{t'=(p,a,\bv,q)\in T}\bv[i]\cdot y_{t'}=0$
    \item $\psi_{2,t}=x_t>0\land y_t>0$
    \item $\psi_{3}=\bigwedge_{t'\in T}(x_{t'}>0 \longleftrightarrow y_{t'}>0)$
  \end{enumerate}
  Additionally, we have to express that $\bmf_1$ and $\bmf_2$ are flows. This is possible with the following formula:
  \[\psi_{0}=\bigwedge_{q\in Q}\sum_{t'=(p,a,\bv,q)\in T}x_{t'}=\sum_{t'=(q,a,\bv,p)\in T}x_{t'}\land \sum_{t'=(p,a,\bv,q)\in T}y_{t'}=\sum_{t'=(q,a,\bv,p)\in T}y_{t'}\,.\]
  Set $\varphi_t=\exists\bx,\by\colon\psi_0\land\psi_1\land\psi_{2,t}\land\psi_3$ where $\bx=(x_{t'})_{t'\in T}$ and $\by=(y_{t'})_{t'\in T}$ are $T$-vectors of variables. Clearly, $\varphi_t$ is satisfiable if, and only if, $t$ is bi-cancelable.
\end{proof}

\section{Construction of the $\Z$-VASS in Section~\ref{section-parikh-automata}}
We only show the construction of $\cW_1$. As described above, $\cW_1$ accepts a sequence\linebreak $\# c_1\# c_2\# \cdots \# c_m$ if $m\in\N$, $c_1,c_2,\ldots,c_m\in S$, and $\Phi(c_1,\ldots,c_m)\in A+U^*+V_J^*$ (where $J=S[\hat{T}]$). This is the case, iff there are vectors $\bu_0\in A+U^*$ and $\bu_1\in V_J^*$ with $\Phi(c_1,\ldots,c_m)=\bu_0+\bu_1$. Recall that $\bu_0\in A+U^*$ is equivalent to $\Delta_{I_1}(\bu_0)+\Delta_{I_2}(\rho)=\bzero$ and that $\bu_1\in V_J^*$ is equivalent to $\Delta_{I_2}(\bu_1)=\bzero$ and $\supp(\bu_1)\in S[\hat{T}]$ (i.e., all transitions in cycles of $\bu_1$ are in $\hat{T}$).

Now, $\cW_1$ is a $|I_1|+|I_2|$-dimensional $\Z$-VASS that will first read a sequence of (short) cycles. For each of these cycles $\cW_1$ guesses whether to count it in $\bu_0$ or $\bu_1$. Accordingly, it adds the effect of each cycle either to the first $|I_1|$ or the last $|I_2|$ counters. In the second case, it also checks the membership of each transition in $\hat{T}$. After reading all the cycles, it finally simulates the skeleton $\rho$ (without reading anything from the input). Since $\Z$-VASS accept with value $0$ in each counter, we will finally obtain $\Delta_{I_1}(\bu_0)+\Delta_{I_2}(\rho)=\bzero$, $\Delta_{I_2}(\bu_1)=\bzero$, and $\supp(\bu_1)\in S[\hat{T}]$.

Recall that $\cV=(Q,\Sigma,T,\iota,f)$ is a $d$-dimensional $\Z$-VASS, $\hat{T}\subseteq T$ is a set of bi-cancelable transitions, and $\rho$ is a skeleton from $\iota$ to $f$ visiting all states in $Q$. We construct a $|I_1|+|I_2|$-dimensional $\Z$-VASS $\cW=(Q',\Gamma,T',\iota,f)$ over the input alphabet $\Gamma=T\cup\{\#\}$ where $\#\notin T$ is a new symbol. The set of states $Q'$ contains (among others) the states $\{\iota,f\}$. We have a transition from $\iota$ to $f$ labeled with $\varepsilon$ and adding $(\Delta_{I_1}(\rho),\bzero)$ to the counters (note that since the skeleton $\rho$ is fixed for our construction, we can simulate it in one step). Additionally, we attach to the state $\iota$ the following two (disjoint) gadgets $\cG_b$ with $b\in\{0,1\}$ simulating short cycles. Here, the index $b$ indicates whether we add the effect of this cycle to the effect of $\bu_0$ or $\bu_1$. Concretely, $\cG_b$ is the following automaton:
\begin{itemize}
  \item the states of $\cG_b$ consist of two states from $Q$ and a bounded counter with values in $[1,|Q|]$, i.e., $\{(p,q,j)\mid p,q\in Q,1\leq j\leq|Q|\}$ is the set of states in $\cG_b$
  \item There are transitions from $\iota$ to each state $(q,q,|Q|)$ with label $\#$ and counter update $(\bzero,\bzero)$. Here, the first state recognizes in which state the simulation of the cycle began, the second one indicates the current state of the simulation, and the counter indicates maximum number of subsequent simulation steps.
  \item For each $1<j\leq|Q|$ we have a transition from $(p,q,j)$ to $(p,q',j-1)$ if $\cV$ has a transition $t=(q,a,\bx,q')\in T$. The label of the new transition is $t$ and the counter update is $(\restrict{\by_0}{I_1},\restrict{\by_1}{I_2})$ where $\by_b=\bx$ and $\by_{1-b}=\bzero$. If $b=1$, we want to simulate twin-unbounded cycles, only. Hence, we also require $t\in\hat{T}$.
  \item We also have transitions from $(p,q,j)$ to $\iota$ if $\cV$ has a transition $t=(q,a,\bx,p)\in T$. The label and the counter update are defined as above.
\end{itemize}
In other words, the gadget $\cG_b$ is actually the computation graph that is truncated to runs of length $\leq|Q|$. Note that each gadget has at most $|Q|^3$ many nodes implying that $\cW$ has polynomial size (in $|Q|$).

The $|I_2|+|I_1|$-dimensional $\Z$-VASS $\cW_2$ is constructed analogously---we only have to replace counter updates $\restrict{\bx}{I_1}$ by $\restrict{\bx}{I_2}$ and vice versa.
Now, we have to show that $L(\cW_1)$ and $L(\cW_2)$ accept the desired languages:

\begin{lemma}\label{app:zvassCorrect1}
  The following equations hold:
  \begin{align*}
    L(\cW_1)&=\{ \#c_1\#c_2\cdots\# c_m \mid m\in\N, c_1,\ldots,c_m\in S, \Phi(c_1,\ldots,c_m)\in A+U^*+V_{S[\hat{T}]}^* \} \\
    L(\cW_2)&=\{ \#c_1\#c_2\cdots\# c_m \mid m\in\N, c_1,\ldots,c_m\in S, \Phi(c_1,\ldots,c_m)\in B+V^*+U_{S[\hat{T}]}^* \}
  \end{align*}
\end{lemma}
\begin{proof}
  We only show the first equation.
  
  Let $w\in L(\cW_1)$ and let $r\in T'^*$ be a $w$-labeled accepting run of $\cW_1$. Clearly, there are $m\in\N$ and words $c_1,\ldots,c_m\in T^*$ with $w=\#c_1\#c_2\cdots\# c_m$. By construction, each $\#c_i$ is read in $r$ by one of the gadgets $\cG_0$ or $\cG_1$. These gadgets simulate runs of length $\leq|Q|$ starting in some state $p\in Q$ that are going back to this state. But these are exactly short cycles, i.e., $c_i\in S$.
  
  Next, for each $1\leq i\leq m$ choose $b_i\in\{0,1\}$ such that in $r$ the factor $\#c_j$ is read via the gadget $\cG_{b_i}$. Let $\bu_0,\bu_1\in\N^S$ be the following two vectors: for each $c\in S$ and $b\in\{0,1\}$ set $\bu_b[c]$ to the number of $1\leq i\leq m$ such that $c=c_i$ and $b=b_i$. Clearly, $\bu_0+\bu_1=\Phi(c_1,\ldots,c_m)$ is exactly the Parikh image of the cycles in $w$. We prove next that $\bu_0\in A+U^*$ and $\bu_1\in V_{S[\hat{T}]}^*$ hold.
  \begin{itemize}
    \item To prove $\bu_0\in A+U^*$ it suffices to show $\Delta_{I_1}(\bu_0)+\Delta_{I_1}(\rho)=\bzero$. We have
      \begin{align*}
           \Delta_{I_1}(\bu_0)+\Delta_{I_1}(\rho)
        &= \sum_{c\in S}\bu_0[c]\cdot\Delta_{I_1}(c)+\Delta_{I_1}(\rho)\\
        &= \sum_{1\leq i\leq m,b_i=0}\Delta_{I_1}(c_i)+\Delta_{I_1}(\rho) && \text{(by definition of $\bu_0$)}\\
        &= \Delta_{I_1}(r) && \text{(by definition of $\cW_1$)}\\
        &= \bzero && \text{(since $r$ is accepting)}
      \end{align*}
    \item We first prove $\Delta_{I_2}(\bu_1)=\bzero$:
      \begin{align*}
           \Delta_{I_2}(\bu_1)
        &= \sum_{c\in S}\bu_1[c]\cdot\Delta_{I_2}(c)\\
        &= \sum_{1\leq i\leq m,b_i=1}\Delta_{I_2}(c_i) && \text{(by definition of $\bu_1$)}\\
        &= \Delta_{I_2}(r) && \text{(by definition of $\cW_1$)}\\
        &= \bzero && \text{(since $r$ is accepting)}
      \end{align*}
      Towards the property $\supp(\bu_1)\subseteq S[\hat{T}]^*$ observe that cycles $c_i$ (with $1\leq i\leq m$) are only counted to $\bu_1$ if $b_i=1$ which is the case if the gadget $\cG_1$ reads this cycle. But $\cG_1$ checks that each transition is in $\hat{T}$ implying that only cycles $c_i\in\hat{T}^*$ are counted to $\bu_1$, i.e.\ $\supp(\bu_1)\subseteq S[\hat{T}]^*$. Finally, from $\Delta_{I_2}(\bu_1)=\bzero$ and $\supp(\bu_1)\subseteq S[\hat{T}]^*$ we infer that $\bu_1\in V_{S[\hat{T}]}^*$ holds.
  \end{itemize}
  Hence, $w$ satisfies all the properties of the right-hand side of the equation.
  
  Towards the converse inclusion, let $m\in\N$, $c_1,\ldots,c_m\in S$, $\Phi(c_1,\ldots,c_m)\in A+U^*+V_{S[\hat{T}]}^*$. We will show $\#c_1\#c_2\cdots\# c_m\in L(\cW_1)$. From $\Phi(c_1,\ldots,c_m)\in A+U^*+V_{S[\hat{T}]}^*$ we obtain the existence of two vectors $\bu_0\in A+U^*$ and $\bu_1\in V_{S[\hat{T}]}^*$ with $\Phi(c_1,\ldots,c_m)=\bu_0+\bu_1$. We construct a run from $\iota$ to $f$ in $\cW_1$ reading  $\#c_1\#c_2\cdots\# c_m$ as follows: for $1\leq i\leq m$ choose a value $b_i\in\{0,1\}$ such that $\bu_b[c]=|\{1\leq i\leq m\mid c_i=c,b_i=b\}|$ holds for all $b\in\{0,1\}$. Let $r_i\in T'^*$ be the (unique) run of $\cG_{b_i}$ with label $\#c_i$. Then $r=r_1r_2\cdots r_mt\in T'^*$ (where $t$ is the transition from $\iota$ to $f$) is a run from $\iota$ to $f$ in $\cW_1$ with label $\#c_1\#c_2\cdots\# c_m$. To show acceptance, we also need that $\Delta(r)=(\bzero,\bzero)$ holds.
  \begin{itemize}
    \item We first show $\Delta_{I_1}(r)=\bzero$:
      \begin{align*}
           \Delta_{I_1}(r)
        &= \sum_{i=1}^m \Delta_{I_1}(r_i) + \Delta_{I_1}(t)\\
        &= \sum_{i=1}^m \Delta_{I_1}(r_i) + \Delta_{I_1}(\rho) && \text{(by definition of $t$)}\\
        &= \sum_{1\leq i\leq m,b_i=0}\Delta_{I_1}(c_i) + \Delta_{I_1}(\rho) && \text{(since $\Delta_{I_1}(c_i)=\bzero$ if $b_i=1$)}\\
        &= \sum_{c\in S} \bu_0[c]\cdot \Delta_{I_1}(c) + \Delta_{I_1}(\rho) && \text{(since $\bu_0[c]=|\{1\leq i\leq m\mid c_i=c,b_i=0\}|$)}\\
        &= \Delta_{I_1}(\bu_0) + \Delta_{I_1}(\rho)\\
        &= \bzero && \text{(since $\bu_0\in A+U^*$)}
      \end{align*}
    \item Now we show $\Delta_{I_2}(r)=\bzero$:
      \begin{align*}
           \Delta_{I_2}(r)
        &= \sum_{i=1}^m \Delta_{I_2}(r_i) + \Delta_{I_2}(t)\\
        &= \sum_{i=1}^m \Delta_{I_2}(r_i) && \text{(by definition of $t$)}\\
        &= \sum_{1\leq i\leq m,b_i=1}\Delta_{I_2}(c_i) && \text{(since $\Delta_{I_2}(c_i)=\bzero$ if $b_i=0$)}\\
        &= \sum_{c\in S} \bu_1[c]\cdot \Delta_{I_2}(c) && \text{(since $\bu_1[c]=|\{1\leq i\leq m\mid c_i=c,b_i=1\}|$)}\\
        &= \Delta_{I_2}(\bu_1)\\
        &= \bzero && \text{(since $\bu_1\in V_{S[\hat{T}]}^*$)}
      \end{align*}
  \end{itemize}
  Hence, the run $r$ is accepting in $\cW_1$ implying $\#c_1\#c_2\cdots\# c_m\in L(\cW_1)$.
\end{proof}

\begin{lemma}\label{app:zvassCorrect2}
  We have $L(\cW_1)\cap L(\cW_2)=\emptyset$ if, and only if, $(A+U^*+V_{S[\hat{T}]}^*)\cap(B+V^*+U_{S[\hat{T}]}^*)=\emptyset$.
\end{lemma}
\begin{proof}
  Assume $L(\cW_1)\cap L(\cW_2)=\emptyset$. Then there is a word $w\in L(\cW_1)\cap L(\cW_2)$. By \cref{app:zvassCorrect1} there are $m\in\N$ and $c_1,\ldots,c_m\in S$ with
  \begin{enumerate}[(i)]
    \item $w=\#c_1\#c_2\cdots \#c_m$,
    \item $\Phi(c_1,\cdots,c_m)\in A+U^*+V_{S[\hat{T}]}^*$, and
    \item $\Phi(c_1,\cdots,c_m)\in B+V^*+U_{S[\hat{T}]}^*$.
  \end{enumerate}
  Hence, we have $\Phi(c_1,\cdots,c_m)\in (A+U^*+V_{S[\hat{T}]}^*)\cap(B+V^*+U_{S[\hat{T}]}^*)\neq\emptyset$.
  
  Conversely, assume $(A+U^*+V_{S[\hat{T}]}^*)\cap(B+V^*+U_{S[\hat{T}]}^*)\neq\emptyset$. Then there is a vector $\bu\in (A+U^*+V_{S[\hat{T}]}^*)\cap(B+V^*+U_{S[\hat{T}]}^*)$. Let $m\in\N$ and $c_1,\ldots,c_m\in S$ be such that $\bu=\Phi(c_1,\ldots,c_m)$. But then \cref{app:zvassCorrect1} yields $\#c_1\#c_2\cdots \#c_m\in L(\cW_1)\cap L(\cW_2)\neq\emptyset$.
\end{proof}

\end{appendix}
\fi

\label{afterbibliography}
\newoutputstream{pagestotal}
\openoutputfile{main.pagestotal.ctr}{pagestotal}
\addtostream{pagestotal}{\getpagerefnumber{afterbibliography}}
\closeoutputstream{pagestotal}

\newoutputstream{todos}
\openoutputfile{main.todos.ctr}{todos}
\addtostream{todos}{\arabic{@todonotes@numberoftodonotes}}
\closeoutputstream{todos}

\end{document}